\documentclass[a4paper]{article} 
\usepackage{graphicx}                   
\usepackage{amsmath}                    
\usepackage{amssymb}                    
\usepackage{amsfonts}                   
\usepackage{url}                        
\usepackage{color}                      
\usepackage{cite}                       
\usepackage{hyperref}                   
\hypersetup{colorlinks=true,linkcolor=[rgb]{0.75,0,0},citecolor=[rgb]{0,0,0.75}}%
\usepackage{wrapfig}                    
\usepackage{thmtools}                   
\usepackage{thm-restate}
\usepackage{fullpage}
\declaretheorem[name=Theorem,numberwithin=section]{thm}
\declaretheorem[name=Lemma,numberwithin=section]{lem}
\declaretheorem[name=Corollary,numberwithin=section]{cor}

\newcommand{\floor}[1]{\left\lfloor #1\right\rfloor}

\newcommand{\RE}{\mathbb{R}}            
\newcommand{\eps}{\varepsilon}          
\newcommand{\half}[1]{\frac{#1}{2}}
\newcommand{\inv}[1]{\frac{1}{#1}}
\newcommand{\radius}{\mathrm{radius}}
\newcommand{\diam}{\mathrm{diam}}
\newcommand{\vol}{\mathrm{vol}}

\newcommand{\width}{\mathrm{width}}
\newcommand{\ray}{\mathrm{ray}}

\newcommand{\etal}{\textit{et al.}}
\newcommand{\SoftOh}{\widetilde{O}}

\begin{document}
\title{Optimal Approximate Polytope Membership}

\author{%
	Sunil Arya\thanks{Research supported by the Research Grants Council of Hong Kong, China under project number 610012.}\\
		Department of Computer Science and Engineering \\
		The Hong Kong University of \\
        Science and Technology \\
 		Clear Water Bay, Kowloon, Hong Kong\\
		arya@cse.ust.hk \\
		\and
	Guilherme D. da Fonseca\\
		Universit\'{e} Clermont Auvergne and\\
		LIMOS\\
 		Clermont-Ferrand, France\\
		fonseca@isima.fr
		\and
	David M. Mount\thanks{Research supported by NSF grant CCF--1618866.}\\
		Department of Computer Science and \\
		Institute for Advanced Computer Studies \\
		University of Maryland \\
 		College Park, Maryland 20742 \\
		mount@cs.umd.edu \\
}

\date{}

\maketitle

\begin{abstract}
In the polytope membership problem, a convex polytope $K$ in $\RE^d$ is given, and the objective is to preprocess $K$ into a data structure so that, given a query point $q \in \RE^d$, it is possible to determine efficiently whether $q \in K$. We consider this problem in an approximate setting and assume that $d$ is a constant. Given an approximation parameter $\eps > 0$, the query can be answered either way if the distance from $q$ to $K$'s boundary is at most $\eps$ times $K$'s diameter. Previous solutions to the problem were on the form of a space-time trade-off, where logarithmic query time demands $O(1/\eps^{d-1})$ storage, whereas storage $O(1/\eps^{(d-1)/2})$ admits roughly $O(1/\eps^{(d-1)/8})$ query time. In this paper, we present a data structure that achieves logarithmic query time with storage of only $O(1/\eps^{(d-1)/2})$, which matches the worst-case lower bound on the complexity of any $\eps$-approximating polytope. Our data structure is based on a new technique, a hierarchy of ellipsoids defined as approximations to Macbeath regions.

As an application, we obtain major improvements to approximate Euclidean nearest neighbor searching. Notably, the storage needed to answer $\eps$-approximate nearest neighbor queries for a set of $n$ points in $O(\log \frac{n}{\eps})$ time is reduced to $O(n/\eps^{d/2})$. This halves the exponent in the $\eps$-dependency of the existing space bound of roughly $O(n/\eps^d)$, which has stood for 15 years (Har-Peled, 2001).
\end{abstract}

\section{Introduction} \label{s:intro}

Convex polytopes are key structures in many areas of mathematics and computation. In this paper, we consider a fundamental search problem related to these objects. Let $K$ denote a convex polytope in $\RE^d$, that is, the bounded intersection of $n$ halfspaces. The \emph{polytope membership problem} is that of preprocessing $K$ so that it is possible to determine efficiently whether a given query point $q \in \RE^d$ lies within $K$. Throughout, we assume that the dimension $d$ is a fixed constant and that $K$ is full dimensional.

It follows from standard results in projective duality that polytope membership is equivalent to answering halfspace emptiness queries for a set of $n$ points in $\RE^d$. In dimension $d \leq 3$, it is possible to build a data structure of linear size that can answer such queries in logarithmic time~\cite{textbook,DoK83}. In higher dimensions, however, the fastest exact data structures with near-linear space have a query time of roughly $O\big(n^{1-1/\floor{d/2}}\big)$~\cite{Mat92}, which is unacceptably high for many applications. 

Polytope membership is a special case of polytope intersection queries \cite{ChD87,DoK83,BaL15}. Recently, Barba and Langerman \cite{BaL15} showed that for any fixed $d$, it is possible to preprocess polytopes in $\RE^d$ so that given two such polytopes that have been translated and rotated, it can be determined whether they intersect each other in time that is logarithmic in their total combinatorial complexity. The preprocessing time and space are quite high, growing as the combinatorial complexity of the polytope (which can be as high as $\Theta(n^{\floor{d/2}})$) raised to the power $\floor{d/2}$.

The lack of efficient exact solutions has motivated consideration of approximate solutions. Let $\eps$ be a positive real parameter, and let $\diam(K)$ denote $K$'s diameter. Given a query point $q \in \RE^d$, an \emph{$\eps$-approximate polytope membership query} returns a positive result if $q \in K$, a negative result if the distance from $q$ to its closest point in $K$ is greater than $\eps \cdot \diam(K)$, and it may return either result otherwise. Polytope membership queries, both exact and approximate, arise in many application areas, such as linear-programming and ray-shooting queries~\cite{Cha96a, Cha10, Ram00, Mat93a, Mat93b}, nearest neighbor searching and the computation of extreme points~\cite{Cha96b,Cla94}, collision detection~\cite{EGS99}, and machine learning~\cite{Bur98}. 

Dudley~\cite{Dud74} showed that, for any convex body $K$ in $\RE^d$, it is possible to construct an $\eps$-approximating polytope $P$ with $O(1/\eps^{(d-1)/2})$ facets. This bound is asymptotically tight in the worst case, even when $K$ is a Euclidean ball. This construction implies a (trivial) data structure for approximate polytope membership problem with space and query time $O(1/\eps^{(d-1)/2})$. Another simple solution arises from the approximation proposed by Bentley {\etal}~\cite{BFP82}. A $d$-dimensional grid with cells of size $\Theta(\eps \cdot \diam(K))$ is created and for every column along the $x_d$-axis, the two extreme $x_d$ values where the column intersects $K$ are stored. Given a query point $q$, it is easy to determine if $q \in P$ in constant time (assuming a model of computation that supports the floor function). The storage required by the approach is $O(1/\eps^{d-1})$.

In~\cite{polytope}, the authors presented a simple and practical data structure for the approximate polytope membership problem, called \emph{SplitReduce}. Given a parameter $t$, space is subdivided hierarchically using a quadtree until each cell either (1) lies entirely inside $K$, (2) entirely outside $K$, or (3) intersects $K$'s boundary and is locally approximable by at most $t$ halfspaces. In the latter case, the leaf node associated with such a cell stores such a set of hyperplanes. To answer a query, the quadtree is descended until arriving at the leaf node whose cell contains the query point. If this node is not labeled as inside or outside, the query is answered by testing whether the query point lies within all the halfspaces stored in the leaf node. In~\cite{polytope} it is shown that the quadtree height is $O(\log\inv{\eps})$, and therefore the overall query time is $O(\log\inv{\eps} + t)$.

A more refined analysis is presented in~\cite{AFM12a}, showing that the minimum storage of $O(1/\eps^{(d-1)/2})$ is attained for query time $t = \Theta\big( (\log\inv{\eps}) / \eps^{(d-1)/8}\big)$. Furthermore, a space-time trade-off is presented that involves a piecewise linear function. Obtaining a tight analysis remains an open problem. A lower-bound proof shows that the storage requirement increases when the query time $t$ drops down to roughly $O(1/\eps^{(d-1)/18})$ \cite{polytope}. Furthermore, the data structure provides no improvement over the storage in~\cite{BFP82} when the query time is polylogarithmic.

While the SplitReduce data structure is both simple and practical, the question of whether it is possible to achieve query time $O(\log\inv{\eps})$ with minimum storage $O(1/\eps^{(d-1)/2})$ has remained open. In this paper, we give an affirmative answer to this question. We abandon the quadtree-based approach of~\cite{polytope} and~\cite{AFM12a} in favor of a data structure involving a hierarchy of ellipsoids. These ellipsoids are selected through a sampling process that is inspired by a classical structure from the theory of convexity, called \emph{Macbeath regions}~\cite{Mac52}. Here is our main result.

\begin{thm} \label{thm:main}
Given a convex polytope $K$ in $\RE^d$ and an approximation parameter $0 < \eps \le 1$, there is a data structure that can answer $\eps$-approximate polytope membership queries with
\[
	\hbox{Query time:~} O\kern-2pt \left( \log \inv{\eps} \right)
	\;\hbox{and}\;
	\textrm{Space:~} O\kern-2pt \left( \frac{1}{\eps^{(d-1)/2}} \right).
\]
\end{thm}

Our focus is on the existence of this data structure. Preprocessing will be discussed in future work, but assuming that $K$ is represented as the intersection of $h$ halfspaces, the construction described in Section~\ref{ss:construct} can be implemented in time $O(n + \mathrm{poly}(1/\eps))$, with polynomial exponents depending on $d$. The principal contribution of this paper is to show that through the use of a more ``shape-sensitive'' approach, it is possible to achieve dramatic improvements over the space requirements of the data structure.

\bigskip

As evidence of the importance of this result, we show that it can be applied to produce significant improvements in the efficiency of \emph{approximate nearest-neighbor searching} in Euclidean space. Approximate nearest neighbor searching in spaces of fixed dimension has been widely studied. Data structures with $O(n)$ storage and query times no better than $O(\log n + 1/\eps^{d-1})$ have been proposed by several authors~\cite{AMNSW98, Bes96, Cha02, DGK01}. In subsequent papers, it was shown that query times could be reduced at the expense of greater storage~\cite{Cha98, Har01, Cla94, SSS06}. Har-Peled introduced the AVD (approximate Voronoi diagram) data structure and showed that $O(\log \frac{n}{\eps})$ query time could be achieved using $\SoftOh(n/\eps^{d})$ space~\cite{Har01}. (The notation $\SoftOh(\cdot)$ conceals logarithmic factors.)

Space-time trade-offs were established for the AVD in a series of papers \cite{AMM02, AVD-JACM, ArM02, AFM10}. At one end of the spectrum, it was shown that with $O(n)$ storage, queries could be answered in time $O(\log n + 1/\eps^{(d-1)/2})$. At the other end, queries could be answered in time $O(\log \frac{n}{\eps})$ with space $\SoftOh(n/\eps^d)$. In~\cite{polytope}, the authors presented a reduction from Euclidean approximate nearest neighbor searching to polytope membership. They established significant improvements to the best trade-offs throughout the middle of the spectrum, but the extremes were essentially unchanged~\cite{polytope,AFM12a}. While the AVD is simple and practical, in \cite{AVD-JACM} lower bounds were presented that imply that significant improvements at the extreme ends of the spectrum are not possible in this model. Through the use of our new data structure for polytope membership, we achieve the following improved trade-off.

\begin{thm} \label{thm:ann-ub}
Given a set $X$ of $n$ points in $\RE^d$, an approximation parameter $0 < \eps \le 1$, and $m$ such that $\log \inv{\eps} \leq m \leq 1/(\eps^{d/2} \log \inv{\eps})$, there is a data structure that can answer Euclidean $\eps$-approximate nearest neighbor queries with
\[
	\hbox{Query time:~} O\kern-2pt \left(\log n + \frac{1}{m \cdot \eps^{d/2}}\right)
	\quad\hbox{and}\]
\[
	\textrm{Space:~} O\kern-2pt \left(n \kern+1pt m \right).
\]
\end{thm}

By setting $m$ to its upper limit it is possible to achieve logarithmic query time while roughly \emph{halving} the exponent in the $\eps$-dependency of the previous best bound, as expressed in the following corollary.

\begin{cor} \label{cor:ann-fast}
Given a set $X$ of $n$ points in $\RE^d$ and an approximation parameter $0 < \eps \le 1$, there is a data structure that can answer Euclidean $\eps$-approximate nearest neighbor queries with
\[
	\hbox{Query time:~} O\kern-2pt \left(\log \frac{n}{\eps} \right)
	\quad\hbox{and}\quad
	\textrm{Space:~} O\kern-2pt \left( \frac{n}{\eps^{d/2}} \right).
\]
\end{cor}

The rest of the paper is organized as follows. In the next section we present definitions and preliminary results. In Section~\ref{s:ds} we present the data structure and analyze its performance. Section~\ref{s:ann} discusses the application to approximate nearest-neighbor searching.

\section{Geometric Preliminaries} \label{s:prelim}

Throughout, we assume that $K$ is presented as the intersection of halfspaces. Note however that our results are largely insensitive to the exact representation or the combinatorial complexity of $K$. (The exceptions are our remarks on the construction of the data structure and choice of hyperplane witnesses to non-membership). For this reason, we will often refer to $K$ simply as a convex body.

It will be convenient to define the approximation error in absolute terms. Given a query point $q \in \RE^d$, an \emph{absolute} $\eps$-approximate polytope membership query returns a positive result if $q \in K$, a negative result if the distance from $q$ to its closest point in $K$ is greater than $\eps$, and it may return either result otherwise. We may assume throughout that $d \ge 4$, since polytope membership queries (which may be applied to the Dudley approximation) can be answered exactly in logarithmic time for $d \le 3$ \cite{DoK83}.

\subsection{Canonical Position and Ray Shooting.}

Let $\partial K$ denote the boundary of $K$. Let $O$ denote the origin of $\RE^d$, and for $x \in \RE^d$ and $r \ge 0$, let $B^r(x)$ denote the Euclidean ball of radius $r$ centered at $x$. Given a parameter $0 < \gamma \le 1$, we say that a convex body $K$ is \emph{$\gamma$-fat} if there exist concentric Euclidean balls $B$ and $B'$, such that $B \subseteq K \subseteq B'$, and $\radius(B) / \radius(B') \ge \gamma$. We say that $K$ is \emph{fat} if it is $\gamma$-fat for a constant $\gamma$ (possibly depending on $d$, but not on $\eps$). 

Let $B_0$ denote a ball of radius $r_0 = 1/2$ centered at the origin. For $0 < \gamma \le 1$, let $\gamma B_0$ denote the concentric ball of radius $\gamma r_0 = \gamma / 2$. We say that a convex body $K$ is in \emph{$\gamma$-canonical form} if its boundary is nested between $\gamma B_0$ and $B_0$ (see Figure~\ref{f:prelim}(a)). A body in $\gamma$-canonical form is $\gamma$-fat, and $\diam(K) \in [\gamma,1]$. We will refer to point $O$ as the \emph{center} of $K$.

\begin{figure*}[htbp]
  \centerline{\includegraphics[scale=.75]{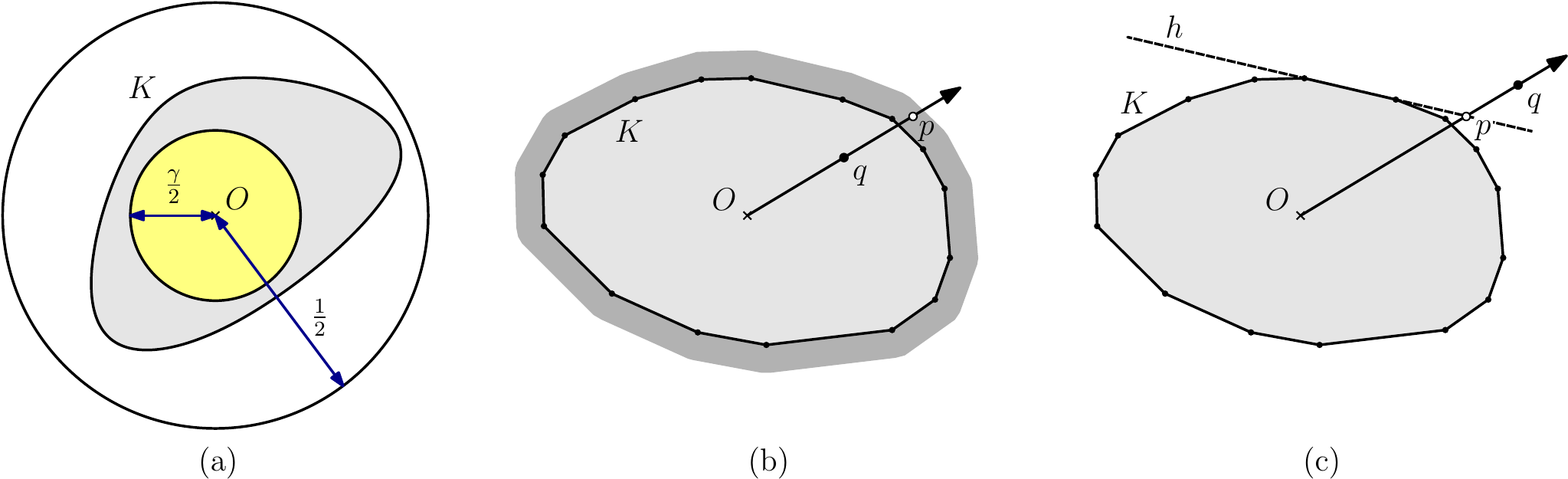}}
  \caption{\label{f:prelim}(a) $\gamma$-canonical form, (b) $\eps$-approximate ray-shooting query, (c) witness.}
\end{figure*}

The next lemma shows that, up to constant factors, the problem of answering relative $\eps$-approximate polytope membership queries can be reduced to the problem of answering absolute $(\eps/d)$-approximate queries with respect to a convex body in $(1/d)$-canonical form. The proof follows from a combination of John's Theorem~\cite{John} and Lemma~{3.1} of Agarwal {\etal}~\cite{AHV04}. (Also, see Lemma~{2.1} of the arXiv version of \cite{AFM16}.)

\begin{lem} \label{lem:canonical}
Let $K \subset \RE^d$ be a convex body. There exists a non-singular affine transformation $T$ such that $T(K)$ is in $(1/d)$-canonical form. Further, if $q$ is a point at distance greater than $\eps \cdot \diam(K)$ from $K$, then $T(q)$ is at distance greater than $\eps/d$ from $T(K)$.
\end{lem}

In light of this result, we may assume henceforth that $K$ is presented in $\gamma$-canonical form, for any constant $\gamma$ (depending on dimension), and that $\eps$ has been appropriately scaled. (This scaling will affect the constant factors hidden in our asymptotic bounds.) Henceforth, we focus on the problem of answering absolute $\eps$-approximate polytope membership queries with respect to $K$.

Our query algorithm solves a slightly more general problem, which will be exploited later in Section~\ref{s:ann}. Given a convex body in $\gamma$-canonical form and any point $q \in \RE^d \setminus \{O\}$, consider the (infinite) ray with origin at $O$ and passing through $q$, which we denote as $Oq$. An \emph{$\eps$-approximate ray shooting query} returns a point $p$ that lies on this ray and is not internal to $K$ but lies within distance $\eps$ of $K$%
\footnote{In light of Lemma~\ref{lem:canonical}, approximate ray-shooting queries also can be defined for an arbitrary convex body. The ray's origin is chosen to be the center of the John ellipsoid and the distance to the point $p$ is relative to $K$'s diameter. In general the ray's central point may be located at any point in $K$'s interior with the property that $K$'s boundary is sandwiched between two uniformly scaled copies of an ellipsoid, both centered at this point. As in Lemma~\ref{lem:canonical}, the value of $\eps$ needs to be adjusted based on the scale factor.}
(see Figure~\ref{f:prelim}(b)). Given the answer to such a ray-shooting query, we can answer approximate membership queries for a query point $q$ by applying the query to the ray $Oq$ and testing whether $q$ lies on the portion of the ray between $O$ and $p$. If so, then (by convexity and the fact that $O$ is interior to $K$) $q$ lies within distance $\eps$ of $K$. If not, $q$ does not lie within $K$. In Section~\ref{s:ds} we will show the following.

\begin{lem} \label{lem:ray-shoot}
Given an arbitrary constant $\gamma$, a convex polytope $K$ in $\RE^d$ that is in $\gamma$-canonical form, and an approximation parameter $0 < \eps \le 1$, there is a data structure that can answer $\eps$-approximate ray-shooting queries in $O(\log \inv{\eps})$ time and $O(1/\eps^{(d-1)/2})$ space.
\end{lem}

Theorem~\ref{thm:main} follows directly from Lemmas~\ref{lem:canonical} and~\ref{lem:ray-shoot}. Our ray-shooting algorithm satisfies the additional property that, when $K$ is given as the intersection of halfspaces, the reported point $p$ lies on the bounding hyperplane $h$ of one of these halfspaces (see Figure~\ref{f:prelim}(c)). The query returns not only $p$ but $h$ as well. As such, if $q$ is reported to lie outside of $K$, then $h$ serves as a witness to $q$'s non-membership. This fact will be exploited in Section~\ref{s:ann}.

\subsection{Caps and Macbeath Regions.}

Much of the material in this section has been presented in \cite{AFM16}. We include it here for the sake of completeness. 

Given a convex body $K$, a \emph{cap} $C$ is defined to be the nonempty intersection of the convex body $K$ with a halfspace (see Figure~\ref{f:cap}(a)). Let $h$ denote the hyperplane bounding this halfspace. We define the \emph{base} of $C$ to be $h \cap K$. The \emph{apex} of $C$ is any point in the cap such that the supporting hyperplane of $K$ at this point is parallel to $h$. The \emph{width} of $C$, denoted $\width(C)$, is the distance between $h$ and this supporting hyperplane. Given any cap $C$ of width $w$ and a real $\rho \ge 0$, we define its \emph{$\rho$-expansion}, denoted $C^{\rho}$, to be the cap of $K$ cut by a hyperplane parallel to and at distance $\rho w$ from this supporting hyperplane. (Note that $C^{\rho} = K$, if $\rho w$ exceeds the width of $K$ along the defining direction.) An easy consequence of convexity is that, for $\rho \ge 1$, $C^{\rho}$ is a subset of the region obtained by scaling $C$ by a factor of $\rho$ about its apex. This implies the following lemma.

\begin{lem} \label{lem:cap-exp}
Let $K \subset \RE^d$ be a convex body and $\rho \ge 1$. For any cap $C$ of $K$, $\vol(C^{\rho}) \le \rho^d \cdot \vol(C)$.
\end{lem}

\begin{figure*}[htbp]
  \centerline{\includegraphics[scale=.75]{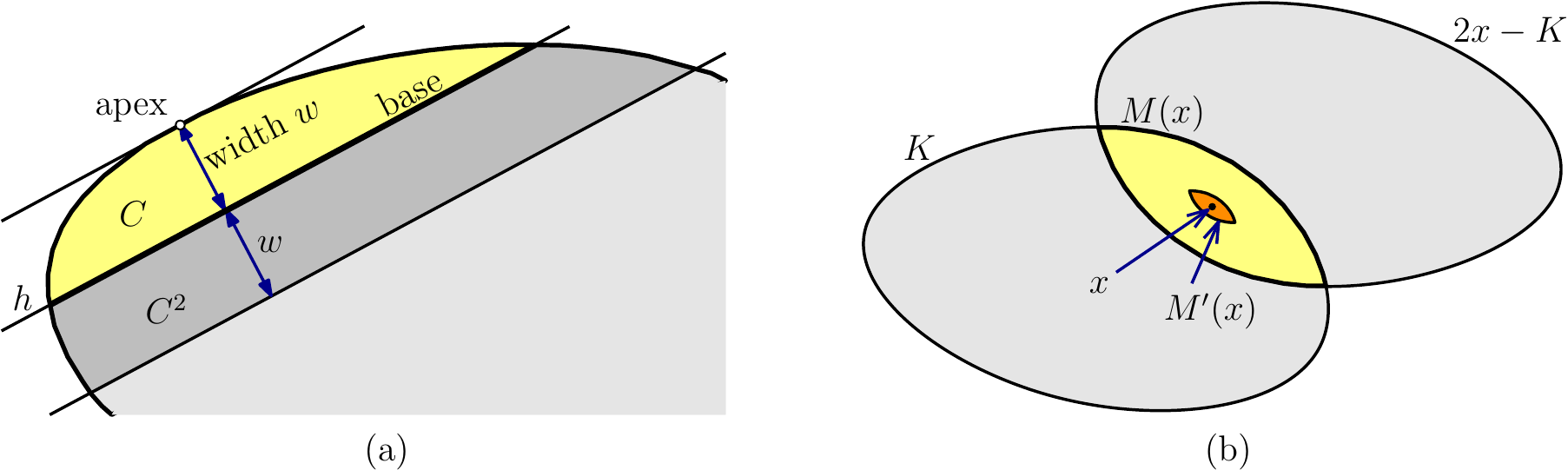}}
  \caption{\label{f:cap}(a) Cap concepts and (b) Macbeath regions.}
\end{figure*}

Given a point $x \in K$ and real parameter $\lambda \ge 0$, the \emph{Macbeath region} $M^{\lambda}(x)$ (also called an \emph{M-region}) is defined as:
\[
  M^{\lambda}(x) ~ = ~ x + \lambda ((K-x) \cap (x-K)).
\]
It is easy to see that $M^{1}(x)$ is the intersection of $K$ and the reflection of $K$ around $x$ (see Figure~\ref{f:cap}(b)), and so $M^{1}(x)$ is centrally symmetric about $x$. $M^{\lambda}(x)$ is a scaled copy of $M^{1}(x)$ by the factor $\lambda$ about $x$. We refer to $x$ as the \emph{center} of $M^{\lambda}(x)$ and to $\lambda$ as its \emph{scaling factor}. As a convenience, we define $M(x) = M^1(x)$ and $M'(x) = M^{1/5}(x)$. We refer to the latter as the \emph{shrunken} Macbeath region.

Macbeath regions have found numerous uses in the theory of convex sets and the geometry of numbers (see B\'{a}r\'{a}ny~\cite{Bar00} for an excellent survey). They have also been applied to a growing number of results in the field of computational geometry, particularly to construct lower bounds for range searching~\cite{BCP93, AMX12, AMM09b} and to bound the complexity of an $\eps$-approximating polytope~\cite{AFM12b,AFM16}.

Given any point $x \in K$, we define a \emph{minimal cap} $C(x)$ to be the cap with minimum volume that contains $x$. Clearly, the base of the minimal cap must pass through $x$. In fact, a standard variational argument implies $x$ is the centroid of the base (otherwise, we could decrease the cap volume by an infinitesimal rotation of the base about $x$~\cite{ELR70}). If the minimal cap is not unique, the notation $C(x)$ will refer to any one of these caps fixed arbitrarily. Define $v(x) = \vol(C(x))$ and $\width(x) = \width(C(x))$. It will be convenient to use $C^{\rho}(x)$ to refer to the $\rho$-expansion of $C(x)$, that is, $C^{\rho}(x) = (C(x))^{\rho}$.

We now present two lemmas that encapsulate key properties of Macbeath regions, which will be useful in the development of our data structure. The first lemma shows that if two shrunken Macbeath regions have a nonempty intersection, then a constant factor expansion of one contains the other~\cite{BCP93,ELR70}. Since the statement we need is slightly different from that proved in earlier papers, we give a proof in the appendix.

\begin{restatable}{lem}{MacMacLemmaStmt}\label{lem:mac-mac}
Let $K$ be a convex body, and let $\lambda \le 1/5$ be any real. If $x,y \in K$ such that $M^{\lambda}(x) \cap M^{\lambda}(y) \neq \emptyset$, then $M^{\lambda}(y) \subseteq M^{4\lambda}(x)$. 
\end{restatable}

The next lemma shows that the minimal cap associated with a point is contained within a suitable constant factor expansion of the associated Macbeath region. It is a straightforward adaptation of a lemma proved by Ewald, Larman and Rogers~\cite{ELR70} (see proof of Lemma~4 in \cite{ELR70}).

\begin{lem} \label{lem:mac2-var}
Let $K \subset \RE^d$ be a convex body in $\gamma$-canonical form, and let $\Delta_0 = \frac{1}{2} (\gamma^2 / (4d))^d$ be a constant. If $x$ is a point in $K$ that lies within distance $\Delta_0$ of $\partial K$, then $C(x) \subseteq M^{3d}(x)$.
\end{lem}

The following lemma is an immediate consequence of the definition of Macbeath region.

\begin{lem} \label{lem:mac1}
Let $K$ be a convex body and $\lambda > 0$. If $x$ is a point in a cap $C$ of $K$, then $M^{\lambda}(x) \cap K \subseteq C^{1+\lambda}$. Furthermore, if $\lambda \le 1$, then $M^{\lambda}(x) \subseteq C^{1+\lambda}$.
\end{lem}

The next lemma is useful in situations when we know that a shrunken Macbeath region partially overlaps a cap of $K$. It allows us to conclude that a constant factor expansion of the cap will fully contain the Macbeath region. The proof appears in \cite{AFM16}.

\begin{lem} \label{lem:cap-mac}
Let $K$ be a convex body. Let $C$ be a cap of $K$ and $x$ be a point in $K$ such that $C \cap M'(x) \neq \emptyset$. Then $M'(x) \subseteq C^2$. 
\end{lem}

\subsection{Relating Distances and Widths.}

In this section we present a number of geometric results demonstrating the relationship between three notions of the distance from a point lying within a convex body to body's boundary. Throughout, let $K$ be a convex body in $\gamma$-canonical form where $\gamma$ is a constant and let $x \in K$. Recall that $\width(x)$ is the width of $x$'s minimum cap. Define $\delta(x)$ to be the minimum distance from $x$ to any point on $\partial K$. For the sake of ray-shooting queries, we define a ray-based notion of distance as well. Given $x \in K$, consider the intersection point $p$ of $\partial K$ and the ray emanating from $O$ and passing through $x$. Define $x$'s \emph{ray-distance}, denoted $\ray(x)$, to be $\|xp\|$ (see Figure~\ref{f:raydeltawidth}).

\begin{figure}[htbp]
  \centerline{\includegraphics[scale=.75]{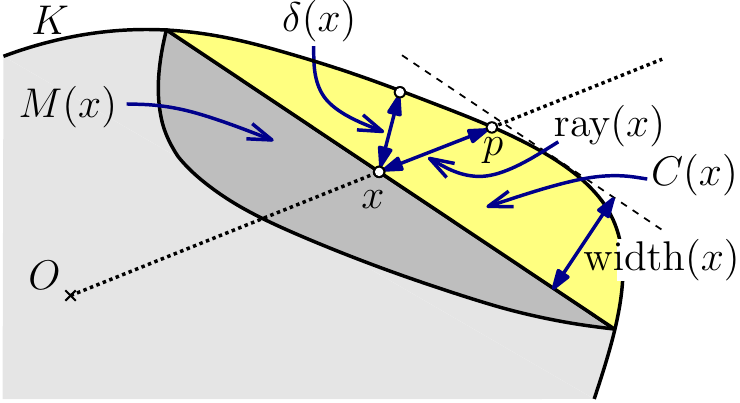}}
  \caption{\label{f:raydeltawidth}Relating $\delta(x)$, $\width(x)$, and $\ray(x)$.}
\end{figure}
 
First we relate $\ray(x)$ and $\delta(x)$. The lower bound on $\ray(x)$ is trivial and the upper bound follows by a straightforward adaptation of Lemma~{4.2} of \cite{AFM16}.

\begin{lem} \label{lem:raydist-delta}
Let $K$ be a convex body in $\gamma$-canonical form. For any point $x \in K$, $\delta(x) \le \ray(x) \le \delta(x) / \gamma$.
\end{lem}

Next, let us relate $\width(x)$ and $\delta(x)$. Clearly, $\width(x) \ge \delta(x)$. In Lemma~\ref{lem:width-delta}, we show that close to the boundary, $\width(x)$ cannot exceed $\delta(x)$ by more than a constant factor. Its proof is based standard properties of Macbeath regions and the following lemma.

\begin{lem} \label{lem:cap-containment}
Let $K$ be a convex body in $\gamma$-canonical form. Let $C_1$ and $C_2$ be two caps of $K$ such that $C_1 \subseteq C_2$. Then $\width(C_1) \le 2 \cdot \width(C_2)/\gamma$.
\end{lem}

\begin{proof}
We consider two cases depending on whether the origin $O$ is inside $C_1$ or not. First, if $O \in C_1$, then $O \in C_2$. Since $K$ contains the ball $B^{\gamma/2}(O)$, it follows that $\width(C_2) \ge \gamma/2$. Since $K$ is contained within the ball $B^{1/2}(O)$, we have $\width(C_1) \le 1$. Thus, $\width(C_1) \le 2 \cdot \width(C_2)/\gamma$. 

Otherwise, we have $O \notin C_1$. Consider the segment joining $O$ to $t$, where $t$ is the apex of $C_1$. Let $x$ denote the point of intersection of this segment with the base of $C_1$. Clearly, $\width(C_1) \le \ray(x)$. By Lemma~\ref{lem:raydist-delta}, $\ray(x) \le \delta(x) / \gamma$. Thus, $\width(C_1) \le \delta(x) / \gamma$. Also, since $x \in C_2$, we have $\delta(x) \le \width(C_2)$. Thus, $\width(C_1) \le \width(C_2)/\gamma$, completing the proof.
\end{proof}

\begin{lem} \label{lem:width-delta}
Let $K \subset \RE^d$ be a convex body in $\gamma$-canon\-ical form, and let $\Delta_0$ be the constant of Lemma~\ref{lem:mac2-var}. If $x$ is a point in $K$ such that $\delta(x) \le \Delta_0$, then $\width(x) \le (2/\gamma) (3d+1) \delta(x)$.
\end{lem}

\begin{proof}
Let $t$ denote the point on $\partial K$ that is closest to $x$. Consider the supporting hyperplane of $K$ at $t$ that is orthogonal to segment $xt$. Consider the halfspace bounded by this hyperplane which does not contain $K$ in its interior. Translate this halfspace such that the bounding hyperplane passes through $x$. Let $C$ denote the cap formed by intersecting this halfspace with $K$. Note that the width of cap $C$ is $\delta(x)$. By Lemma~\ref{lem:mac1}, $M^{3d}(x) \cap K \subseteq C^{3d+1}$. Since $\delta(x) \le \Delta_0$, it follows from Lemma~\ref{lem:mac2-var} that $C(x) \subseteq M^{3d}(x)$. By definition, $C(x) \subseteq K$, so we have
\[
	C(x) 
		~ \subseteq ~ M^{3d}(x) \cap K 
		~ \subseteq ~ C^{3d+1}.
\]
By Lemma~\ref{lem:cap-containment}, it follows that
\begin{eqnarray*}
	\width(x) 
		&  =  & \width(C(x)) \\
        & \le & \frac{2}{\gamma} \kern+1pt \width(C^{3d+1})
		~  =  ~ \frac{2}{\gamma} (3d+1) \delta(x), 
\end{eqnarray*}
as desired.
\end{proof}

The following lemma, illustrated in Figure~\ref{f:projection}, will be useful to analyze the ray shooting performed by our data structure.

\begin{figure}[htbp]
  \centerline{\includegraphics[scale=.75]{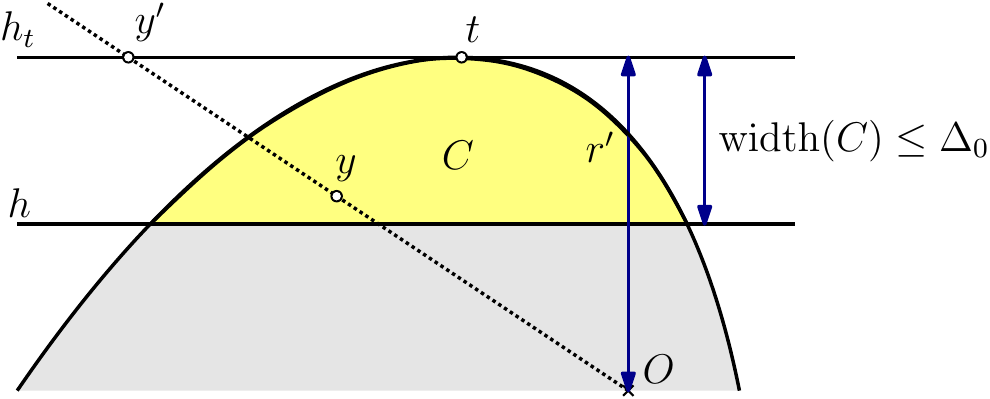}}
  \caption{\label{f:projection}Statement of Lemma~\ref{lem:projection}.}
\end{figure}

\begin{lem} \label{lem:projection}
Let $K$ be a convex body in $\gamma$-canonical form, and let $\Delta_0$ be the constant of Lemma~\ref{lem:mac2-var}. Let $C$ be a cap of width at most $\Delta_0$ defined by a hyperplane $h$, and let $y$ be any point in $C$. Let $t$ be $C$'s apex, and let $h_t$ be the hyperplane parallel to $h$ that passes through $t$. Letting $y'$ denote the intersection of line $Oy$ and $h_t$, we have $\|y y'\| \le 2 \cdot \width(C)/\gamma$.
\end{lem}

\begin{proof}
Given that $y \in C$, $\|y y'\|$ is maximized when $y$ lies on $C \cap h$, and so let us assume this. Since $K$ is in $\gamma$-canonical form, it is nested between two balls of radii $r = \gamma/2$ and $R = 1/2$ centered at $O$. Let $r'$ denote the perpendicular distance from $O$ to $h_t$. Clearly, $h_t$ is a supporting hyperplane of $K$, and so $r' \ge r$. By definition of $\Delta_0$ and since $\gamma \le 1$, we have $\Delta_0 \le \gamma/4 = r/2$. Let $R' = \|O y\|$. Since $y \in K$, $R' \le R$. Letting $w = \width(C)$, by similar triangles we have $R'/(r' - w) = \|y y'\|/w$. Therefore,
\begin{eqnarray*}
	\|y y'\|
	&  =  & \frac{R'}{r' - w} w
	~ \le ~ \frac{R'}{r' - \Delta_0} w
	~ \le ~ \frac{R'}{r - (r/2)} w \\
	& \le & \frac{R}{r/2} w
	~  =  ~ \frac{2 \cdot \width(C)}{\gamma},
\end{eqnarray*}
as desired.
\end{proof}

Finally, we establish a monotonicity relationship between $\delta(x)$ and $\ray(x)$ that holds close to the boundary. For any $\delta > 0$, define the \emph{$\delta$-erosion} of $K$, denoted $K(\delta)$, to be the closed convex body formed by removing from $K$ all points lying within distance $\delta$ of $\partial K$. We can define $K(\delta)$ equivalently as follows. Let $\mathcal{H}$ denote the set of supporting halfspaces of $K$, so that $K = \bigcap_{H \in \mathcal{H}} H$. Letting $\mathcal{H}(\delta)$ denote the set of halfspaces obtained by translating each halfspace of $\mathcal{H}$ towards $O$ by $\delta$, we have $K(\delta) = \bigcap_{H \in \mathcal{H}(\delta)} H$. Recalling that $B^{\gamma/2}(O) \subseteq K$, the next lemma follows from elementary geometry.

\begin{lem} \label{lem:monotone}
Let $K$ be a convex body in $\gamma$-canonical form. The following hold:
\begin{enumerate}
\item[(a)] if $\delta < \gamma / 2$, then $O \in K(\delta)$.

\item[(b)] Consider any ray emanating from $O$. Let $x$ and $y$ denote the points of intersection of this ray with the boundaries of $K(\gamma / 2)$ and $K$, respectively. As point $p$ moves along this ray from $x$ to $y$, $\delta(p)$ decreases strictly monotonically.
\end{enumerate}
\end{lem}

\subsection{Further Properties of Macbeath Regions.}

Finally, we identify some useful novel properties of Macbeath regions. The first lemma is a useful utility. Lemma~\ref{lem:core-delta} shows that all the points in a shrunken Macbeath region have similar distances from the boundary of $K$, and Lemma~\ref{lem:core-vol} shows that the minimal caps associated with these points have similar volumes.

\begin{lem} \label{lem:mac3}
Let $K$ be a convex body. If $x \in K$ and $x' \in M'(x)$, then $x \in M^{1/4}(x')$.
\end{lem}

\begin{proof}
Recalling that $M'(x) = M^{1/5}(x)$, it follows that there exist points $p_1,p_2 \in K$ such that $x' = x + \frac{1}{5}(p_1 - x)$ and $x' = x + \frac{1}{5}(x - p_2)$. After simple algebraic manipulations, the first equation is equivalent to
\begin{equation}
	x ~ = ~ x' + \frac{1}{4} (x' - p_1). \label{eqn:mac3a}
\end{equation}
Letting $p_3 = \frac{2}{3} p_2 + \frac{1}{3} x'$, the second equation is equivalent to
\begin{equation}
	x
	~ = ~ x' + \frac{1}{4} \left( \frac{2}{3} p_2 + \frac{1}{3} x' - x' \right)
	~ = ~ x' + \frac{1}{4} \left(p_3 - x' \right). \label{eqn:mac3b}
\end{equation}
As $p_3$ is a convex combination of $p_2$ and $x'$, we have $p_3 \in K$. Eq.~(\ref{eqn:mac3a}) shows that $x \in x' + (1/4) (x' - K)$, and Eq.~(\ref{eqn:mac3b}) shows that $x \in x' + (1/4) (K - x')$. Thus $x \in M^{1/4}(x')$.
\end{proof}

\begin{lem} \label{lem:core-delta}
Let $K$ be a convex body. If $x \in K$ and $x' \in M'(x)$, then $4 \delta(x) / 5 \le \delta(x') \le 4 \delta(x) / 3$.
\end{lem}

\begin{proof}
To prove the lower bound on $\delta(x')$, let $z$ denote the point of $\partial K$ that is closest to $x'$, and let $h$ be a supporting hyperplane passing through $z$ (see Figure~\ref{f:mac3}). Let $\ell$ denote the (perpendicular) distance from $x$ to $h$, and let $h'$ be the translate of $h$ by distance $4\ell/5$ towards $x$. Because $M(x)$ lies entirely within the halfspace bounded by $h$ that contains the origin, it follows that $M'(x)$ lies entirely within the corresponding halfspace bounded by $h'$. This implies that $\delta(x') \ge 4\ell/5$. Clearly, $\delta(x) \le \ell$, and hence $\delta(x') \ge 4\ell/5 \ge 4 \delta(x)/5$.

\begin{figure}[htbp]
  \centerline{\includegraphics[scale=.75]{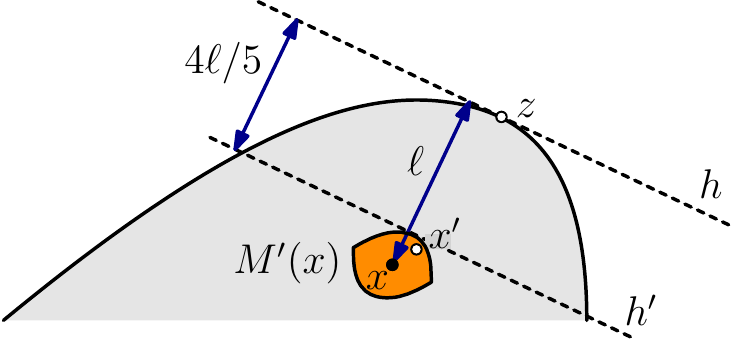}}
  \caption{\label{f:mac3}Proof of Lemma~\ref{lem:core-delta}.}
\end{figure}

To prove the upper bound on $\delta(x')$ observe that, by Lemma~\ref{lem:mac3}, $x \in M^{1/4}(x')$. A symmetrical argument to the above shows that $\delta(x) \ge 3 \delta(x') / 4$, as desired.
\end{proof}

Recall that $C(x)$ is the cap of minimum volume that contains $x$ and $v(x) = \vol(C(x))$.

\begin{lem} \label{lem:core-vol}
Let $K \subset \RE^d$ be a convex body. If $x \in K$ and $x' \in M'(x)$, then $2^d v(x) \ge v(x') \ge v(x) / 2^d$.
\end{lem}

\begin{proof}
By Lemma~\ref{lem:mac1}, $M'(x) \subseteq C^{6/5}(x)$. Therefore, $x' \in C^{6/5}(x)$, implying that the minimum volume cap containing $x'$ has volume at most $\vol(C^{6/5}(x))$. By Lemma~\ref{lem:cap-exp}, $\vol(C^{6/5}(x)) \le (6/5)^d \vol(C(x))$. Thus
\[
	v(x')
	~ \le ~ \vol(C^{6/5}(x))
	~ \le ~ \left( \frac{6}{5} \right)^{\kern-2pt d} v(x)
	~ \le ~ 2^d v(x),
\]
which proves the first inequality. To prove the second inequality observe that, by Lemma~\ref{lem:mac3}, $x \in M^{1/4}(x')$. Arguing as in the proof of the first inequality, we obtain $v(x) \le 2^d v(x')$, as desired.
\end{proof}
 
\section{The Data Structure} \label{s:ds}

Recall that we are given a convex polytope $K \subset \RE^d$ in $\gamma$-canonical form, where $\gamma$ is a constant, and our objective is to construct a data structure that can answer $\eps$-approximate ray-shooting queries. Our approach is to compute a series of nested rings within $K$, each of which surrounds the origin. Each ring is the union of a collection of appropriately scaled Macbeath regions such that any ray shot from the origin hits at least one Macbeath region from each ring (see Figure~\ref{f:dag}). The rings extend outwards towards the boundary of $K$. To simplify query processing, we will replace each Macbeath region with a containing ellipsoid whose volume is larger by at most a constant factor. With each successive level these ``Macbeath ellipsoids'' define successively better approximations to $ \partial K$, with the last ring forming an $\eps$-approximation to $\partial K$. 

\begin{figure}[htbp]
  \centerline{\includegraphics[scale=.75]{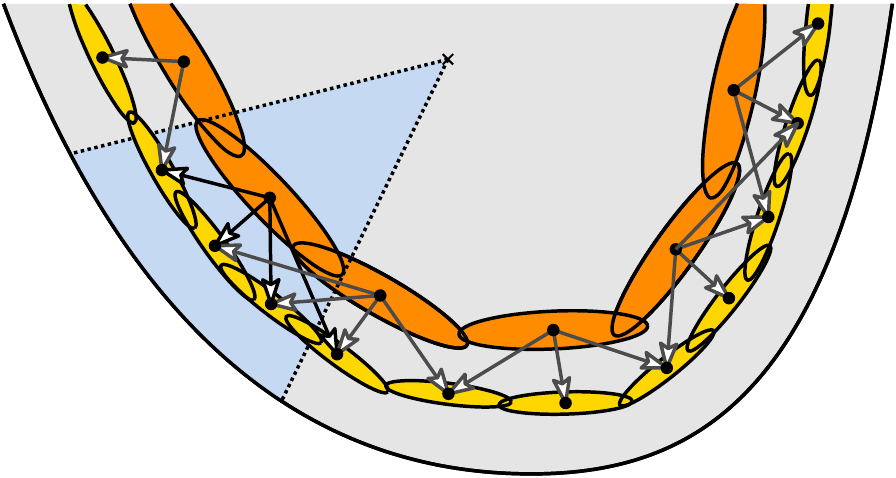}}
  \caption{\label{f:dag}Illustration of two levels of the data structure.}
\end{figure}

These rings naturally define a layered DAG structure whose nodes correspond to Macbeath ellipsoids. A Macbeath ellipsoid at level $i$ is the child of a Macbeath ellipsoid at level $i-1$ if there is a ray from the origin that intersects both of them. (It will in fact hit the ellipsoid at level $i-1$ before the one at level $i$.) We will show that each ellipsoid has a constant number of children, and that the overall depth of this DAG is $O(\log \inv{\eps})$. 

To define the structure more formally, let $\Delta_0$ be the constant of Lemma~\ref{lem:mac2-var}, and for $i \ge 0$ define $\Delta_i = \Delta_0 / 2^i$. The levels are indexed from $0$ to $\ell$, where $\ell$ is the smallest integer such that $\Delta_{\ell} \le \gamma^2 \eps /(8 (3d+1))$. Since $\gamma$ is a constant, $\ell = O(\log \inv{\eps})$. Recall that $K(\delta)$ denotes the body that results by eroding $K$ by distance $\delta$, and let $\lambda_0 = 1/(20\sqrt{d})$ be a constant. By Lemma~\ref{lem:monotone}(a), $K(\Delta_0)$ contains the origin $O$ and $K(\Delta_i) \subset K(\Delta_{i+1})$. The nodes at level $i$ of our data structure correspond to a maximal set of disjoint Macbeath regions $M^{\lambda_0}(x)$ whose centers $x$ lie on the boundary of the eroded body $K(\Delta_i)$. For any node $u$, let $x_u$ denote the center of the associated Macbeath region $M^{\lambda_0}(x_u)$. Define the associated \emph{Macbeath ellipsoid}, denoted $E(x_u)$, to be the circumscribing John ellipsoid of $M^{4\lambda_0}(x_u)$. (Since $M^{\lambda_0}(x_u)$ is centrally symmetric about $x_u$, $E(x_u)$ will be centered about this point.) We will show that the union of the Macbeath ellipsoids at level $i$ cover $\partial K(\Delta_i)$, implying that any ray emanating from the origin must intersect at least one ellipsoid of each level.

As mentioned above, given nodes $u$ and $v$ from levels $i$ and $i+1$, respectively, $v$ is a child of $u$ if there exists a ray emanating from the origin that intersects both $E(x_u)$ and $E(x_v)$. We can \emph{root} the DAG by creating a special node whose children are all the nodes of level zero. In order to produce a witness for approximate ray-shooting queries, we associate each leaf node with a constant number of supporting hyperplanes of $K$ that locally approximate the boundary of $K$ near the leaf's Macbeath ellipsoid. (This will be discussed in detail in Section~\ref{ss:construct}).

Given a ray $Oq$, the query algorithm descends the DAG by starting at the root and visiting any node at level zero that intersects the ray. Letting $u$ denote the current node, we next visit any child of $u$ whose associated ellipsoid intersects the ray. (Such a child must exist.) Upon reaching the leaf level we intersect $Oq$ with all of its associated supporting hyperplanes and return the intersection point $p$ that is closest to $O$ as the answer to the query (along with the identity of the hyperplane containing $p$).

In the subsections below, we present a formal analysis of the structure and its properties. In Section~\ref{ss:construct} we sketch its construction. In Section~\ref{ss:children} we show that each node has $O(1)$ children. In Section~\ref{ss:storage}, we show that the total storage required is $O(1/\eps^{(d-1)/2})$. Finally, in Section~\ref{ss:query} we show that the query algorithm is correct and has query time $O(\log \inv{\eps})$.

\subsection{Construction.} \label{ss:construct}

Since our focus is on the existential properties of the data structure, we will discuss its construction only at  a high level. We are given the convex body $K$ and approximation parameter $\eps$. Due to the approximate nature of the queries, most of the steps can be implemented approximately subject to a suitable adjustment of the constant factors.

The construction begins by converting $K$ into ca\-nonical form as described in Lemma~\ref{lem:canonical}. Next, for $0 \le i \le \ell$, the eroded bodies $K(\Delta_i)$ are computed. Recalling the constant $\lambda_0$ earlier, for each body $K(\Delta_i)$ we greedily compute a maximal set of points $X_i$ on its boundary such that the Macbeath regions $M^{\lambda_0}(x)$ for $x \in X_i$ are pairwise disjoint. For each point $x \in X_i$, we construct the associated Macbeath region $M^{4\lambda_0}(x)$ and the associated Macbeath ellipsoid $E(x)$. We also create a node for this point at level $i$ of the DAG. Finally, for each pair of nodes at consecutive levels of the DAG, we determine whether there exists a ray emanating from the origin that intersects both of their associated Macbeath ellipsoids. If so, we create a parent-child link between them. We create a special root node, which we connect to all the nodes of level zero. This defines the layered DAG structure.

Next, let us consider the assignment of supporting hyperplanes to the leaves of the data structure. Let $u$ be a leaf node, and let $E(x_u)$ denote the associated Macbeath ellipsoid with center point $x_u$ (see Figure~\ref{f:leaf-node}). Let $C(x_u)$ denote the corresponding minimum volume cap. Let $t$ be the apex of this cap, and let $h_t$ denote the hyperplane (which is a supporting hyperplane of $K$) passing through $t$ and parallel to the base of the cap. In Lemma~\ref{lem:query-correctness}, we will show that $h_t$ can serve as the desired witness, but in some applications it is desirable that the witness be chosen from $K$'s bounding hyperplanes. By Carath\'eodory's theorem \cite{Egg58}, there is a set of at most $d$ of $K$'s bounding halfspaces whose intersection defines an unbounded simplex that contains $K$, and this simplex is contained within the halfspace bounded by $h_t$ containing $K$ (shaded in blue in Figure~\ref{f:leaf-node}). The leaf node $u$ stores this set of hyperplanes, which we denote by $H_u$.

\begin{figure}[hbtp]
  \centerline{\includegraphics[scale=.75]{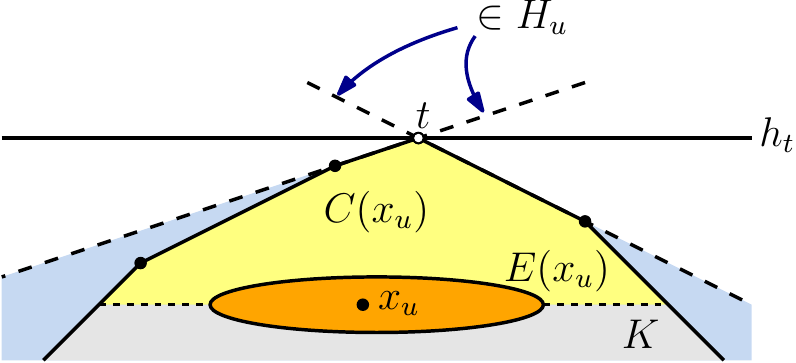}}
  \caption{\label{f:leaf-node}A leaf node in the data structure.}
\end{figure}

\subsection{Bounding the Out-degree.} \label{ss:children}

In this section we show that each node has $O(1)$ children. Intuitively, this involves showing that the set of rays emanating from the origin that pass through a Macbeath ellipsoid for a point on the boundary of $K(\Delta_i)$ can intersect at most a constant number of Macbeath ellipsoids for points on the boundary $K(\Delta_{i+1})$. This is because the points $x$ defining the nodes of each level have disjoint Macbeath regions $M^{\lambda_0}(x)$, which permits us to employ a packing argument.

For any point $x \in K$, recall that $v(x)$ denotes the volume of the minimal cap $C(x)$. Our first lemma considers how $v(x)$ changes as the point $x$ moves towards the boundary of $K$ along a ray emanating from $O$. The lemma shows that if the distance to the boundary, $\delta(x)$, decreases by at most a constant factor, then $v(x)$ decreases by no more than some constant factor. 

\begin{lem} \label{lem:vol-delta}
Let $K \subset \RE^d$ be a convex body in $\gamma$-canonical form. Let $y$ be a point on the ray $Ox$, such that $\ray(y) \le \ray(x)$. If $\delta(y) \ge \delta(x) / \alpha$ for any $\alpha \ge 1$, then $v(y) \ge (\gamma / \alpha)^d \, v(x)$.
\end{lem}

\begin{proof}
If $C(y)$ contains $O$ then, by convexity, it would follow that $x \in C(y)$. This would imply that $v(y) = \vol(C(y)) \ge v(x)$, which would prove the lemma. We may assume therefore that $C(y)$ does not contain $O$.

\begin{figure}[hbtp]
  \centerline{\includegraphics[scale=.75]{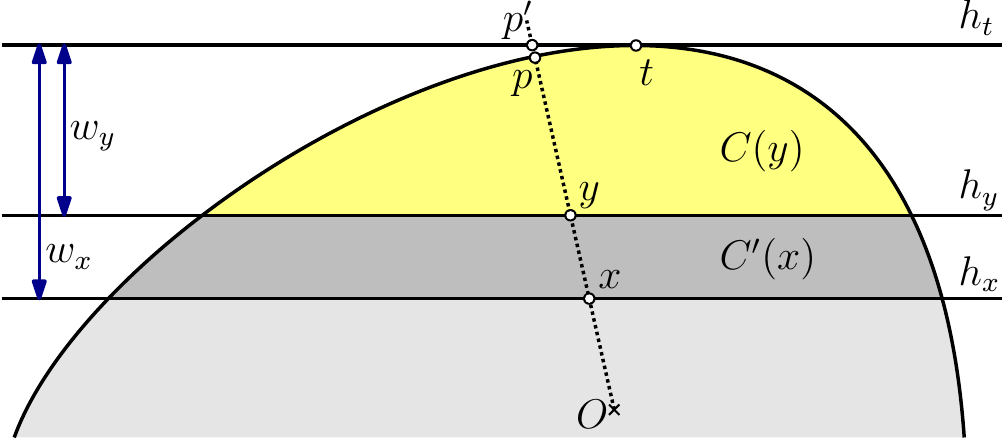}}
  \caption{\label{f:vol-delta}Proof of Lemma~\ref{lem:vol-delta}.}
\end{figure}

 Let $h_y$ denote the hyperplane passing through the base of $C(y)$, and let $t$ denote the apex of $C(y)$. Let $h_t$ and $h_x$ denote the hyperplanes parallel to $h_y$ passing through $t$ and $x$, respectively. Note that $h_t$ is a supporting hyperplane of $K$. Let $C'(x)$ denote the (not necessarily minimal) cap with apex $t$, whose base lies on $h_x$. Let $w_y$ and $w_x$ denote the widths of the caps $C(y)$ and $C'(x)$, respectively. Clearly, $C'(x)$ is a $(w_x/w_y)$-expansion of the cap $C(y)$, and so by Lemma~\ref{lem:cap-exp}, $\vol(C'(x)) \le (w_x/w_y)^d \cdot \vol(C(y))$. Thus
\begin{equation} \label{eq:vx}
 v(x)
	~ \le ~ \vol(C'(x)) 
	~  \le  ~ \left(\frac{w_x}{w_y}\right)^{\kern-2pt d} v(y).
\end{equation}

Next we show that $w_y$ is not much smaller than $w_x$. Let $p$ and $p'$ denote the points of intersection of the ray $Ox$ with $\partial K$ and $h_t$, respectively. Using elementary geometry and the facts that $\ray(y) \ge \delta(y)$ and $\ray(x) \le \delta(x) / \gamma$ (Lemma~\ref{lem:raydist-delta}), we obtain
\begin{eqnarray*}
  \frac{w_x}{w_y} 
	&  =  & \frac{\|xp'\|}{\|yp'\|} 
	~  =  ~ \frac{\ray(x) + \|pp'\|}{\ray(y) + \|pp'\|}\\
	& \le & \frac{\ray(x)}{\ray(y)}
	~ \le ~ \frac{\delta(x) / \gamma}{\delta(y)}
	~ \le ~ \frac{\alpha}{\gamma}.
\end{eqnarray*}
Substituting this bound in Equation~\ref{eq:vx}, we obtain $v(x) \le (\alpha/\gamma)^d \, v(y)$, which completes the proof.
\end{proof}

The following lemma relates the Macbeath regions associated with a node and any of its children.

\begin{lem} \label{lem:degree1}
Let $K \subset \RE^d$ be a convex body in $\gamma$-canonical form for some constant $\gamma$, and let $\Delta_0$ be the constant of Lemma~\ref{lem:mac2-var}. Let $x$ be a point within distance at most $\Delta_0$ of the boundary of $K$. Consider the generalized cone formed by rays emanating from the center $O$ of $K$ and intersecting $M'(x)$. Consider any Macbeath region $M'(y)$ that overlaps this cone where $\delta(y) = \delta(x) / 2$. Then
\begin{enumerate}
\item[(a)] $M'(y) \subseteq C^4(x)$, and

\item[(b)] There exists a constant $c$ (depending on $d$ and $\gamma$) such that $\vol(M(y)) \ge v(x)/c$.
\end{enumerate}
\end{lem}

\begin{figure}[btp]
  \centerline{\includegraphics[scale=.75]{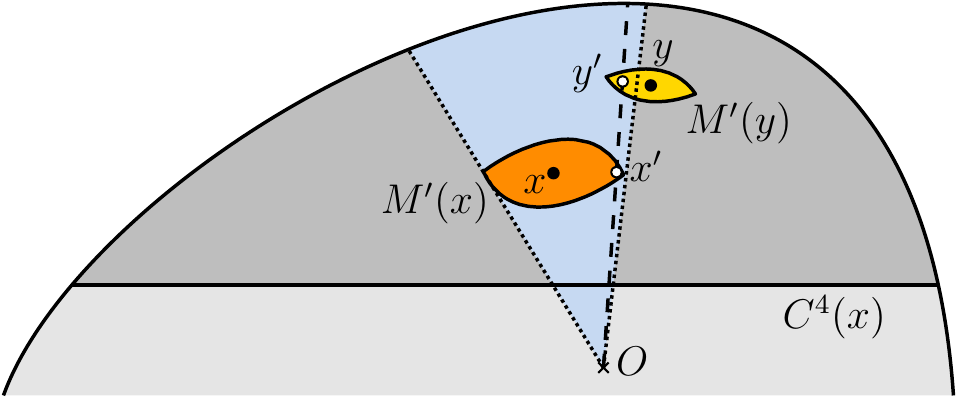}}
  \caption{\label{f:degree1}Proof of Lemma~\ref{lem:degree1}.}
\end{figure}

\begin{proof}
We claim that $M'(y)$ overlaps $C^2(x)$. By Lemma~\ref{lem:cap-mac}, this will imply that $M'(y) \subseteq C^4(x)$ and so will establish (a). To see the claim, consider any ray that emanates from $O$ and intersects both $M'(x)$ and $M'(y)$. Let $x'$ and $y'$ be any two points on this ray that are contained in $M'(x)$ and $M'(y)$, respectively (see Figure~\ref{f:degree1}). Applying Lemma~\ref{lem:core-delta} to points $x$ and $x'$, we obtain $\delta(x') \ge 4 \delta(x) / 5$. Applying the same lemma to points $y$ and $y'$, we obtain $\delta(y') \le 4 \delta(y) / 3$. Recalling that $\delta(y) = \delta(x) / 2$ and putting this all together, we obtain
\[
  \delta(y')
	~ \le ~ \frac{4}{3} \cdot \delta(y)
	~  =  ~ \frac{4}{3} \cdot \frac{\delta(x)}{2} 
	~ \le ~ \frac{4}{3} \cdot \frac{1}{2} \cdot \frac{5}{4} \cdot \delta(x') 
	~  <  ~ \delta(x').
\]

Applying Lemma~\ref{lem:core-delta} to points $x$ and $x'$, we have $\delta(x') \le 4 \delta(x) / 3 \le 4 \Delta_0 / 3$. Substituting the value of $\Delta_0$, it is easy to verify that $\delta(x') < \gamma / 2$. Since $\delta(y') < \delta(x')$, we can now apply Lemma~\ref{lem:monotone}(b) to conclude that $\ray(y') < \ray(x')$. In other words, $y'$ occurs after $x'$ along the ray emanating from $O$. Also, by Lemma~\ref{lem:mac1}, we have $M'(x) \subseteq C^{6/5}(x) \subseteq C^2(x)$. Therefore, $x' \in C^2(x)$, and so $y' \in C^2(x)$. Thus, we have shown that $M'(y)$ intersects $C^2(x)$, which proves (a).

Next we prove (b). Applying Lemma~\ref{lem:core-delta} to points $y$ and $y'$, we obtain $\delta(y') \ge 4 \delta(y) / 5$. Recalling that $\delta(x') \le 4 \delta(x) / 3$, we have
\begin{eqnarray*}
  \delta(y')
	& \ge & \frac{4}{5} \cdot \delta(y) 
	~  =  ~ \frac{4}{5} \cdot \frac{\delta(x)}{2} \\
	& \ge & \frac{4}{5} \cdot \frac{1}{2} \cdot \frac{3}{4} \cdot \delta(x')
	~ \ge ~ \frac{1}{4} \cdot \delta(x').
\end{eqnarray*}
Applying Lemma~\ref{lem:vol-delta} to $x'$ and $y'$, we obtain $v(y') \ge (\gamma/4)^d v(x')$.

Applying Lemma~\ref{lem:core-vol} to $x$ and $x'$, we have $v(x') \ge v(x) / 2^d$. Analogously, we have $v(y) \ge v(y') / 2^d$. Also, since $\delta(y) = \delta(x)/2 \le \Delta_0/2 \le \Delta_0$, the precondition of Lemma~\ref{lem:mac2-var} is satisfied for point $y$. Applying Lemma~\ref{lem:mac2-var}, it follows that $C(y) \subseteq M^{3d}(y)$. Thus 
\[
	\vol(M(y)) 
		~ \ge ~ \frac{\vol(C(y))}{(3d)^d}
		~   = ~ \frac{v(y)}{(3d)^d}. 
\]
Putting it all together, we obtain
\begin{eqnarray*}
  \vol(M(y))
	& \ge & \frac{v(y)}{(3d)^d} 
	~ \ge ~ \frac{1}{(3d)^d} \cdot \frac{1}{2^d} \cdot v(y') \\
	& \ge & \frac{1}{(3d)^d} \cdot \frac{1}{2^d} \cdot \left(\frac{\gamma}{4}\right)^d v(x') \\
	& \ge & \frac{1}{(3d)^d} \cdot \frac{1}{2^d} \cdot \left(\frac{\gamma}{4}\right)^d \cdot \frac{1}{2^d} \cdot v(x) \\
	& \ge & \left( \frac{\gamma}{48 \kern+1pt d} \right)^d \cdot v(x).
\end{eqnarray*}
This yields $\vol(M(y)) \ge v(x) / c$ for any constant $c \ge (48 \kern+1pt d / \gamma)^d$, which proves~(b).
\end{proof}

The previous lemma implies the following.

\begin{lem} \label{lem:degree}
Let $K \subset \RE^d$ be a convex body, and let $\Delta_0$ be the constant of Lemma~\ref{lem:mac2-var}. Also, let $\lambda \le 1/5$ be any constant. Let $x \in K$ such that $\delta(x) \le \Delta_0$. Consider the generalized cone formed by rays emanating from the center $O$ of $K$ and intersecting $M'(x)$. Let $Y$ denote any set of points $y$ such that $\delta(y) = \delta(x)/2$ and the set of Macbeath regions $M^{\lambda}(y)$ are disjoint. Let $Y' \subseteq Y$ denote the set of points $y$ such that $M'(y)$ overlaps the aforementioned cone. Then $|Y'| = O(1)$.
\end{lem}

\begin{proof}
Let $y$ denote any point of $Y'$. Applying Lemma~\ref{lem:degree1}, it follows that (a) $M'(y) \subseteq C^4(x)$, and (b) $\vol(M(y)) \ge v(x)/c$, for a suitable constant $c$. Since $\lambda \le 1/5$, it follows that $M^{\lambda}(y)$ is contained in $C^4(x)$. By Lemma~\ref{lem:cap-exp}, the volume of $C^4(x)$ is at most $4^d v(x) = O(v(x))$ and the volume of $M^{\lambda}(y)$ is $\lambda^d \cdot \vol(M(y)) \ge \lambda^d \cdot v(x) / c = \Omega(v(x))$. Since the Macbeath regions $M^{\lambda}(y)$ for $y \in Y'$ are disjoint, by a straightforward packing argument, it follows that $|Y'| = O(1)$.
\end{proof}

We are now ready to show that the number of children of any non-root node $u$ in our data structure is $O(1)$. (We will analyze the number of children of the root node later. See the remarks following Lemma~\ref{lem:bound-mac}.) Consider any node $u$ at level $i \ge 0$. Recall that $E(x_u)$ denotes the associated Macbeath ellipsoid, which encloses $M^{4 \lambda_0}(x_u)$. The children of $u$ are those nodes $v$ at level $i+1$ whose ellipsoid $E(x_v)$ intersects the generalized cone formed by rays emanating from the origin that intersect $E(x_u)$. The child condition is expressed in terms of Macbeath ellipsoids (for the sake of efficient query processing), but the above lemma is stated in terms of Macbeath regions.

Since $x_u \in \partial K(\Delta_i)$, we have $\delta(x_u) = \Delta_i \le \Delta_0$. Macbeath regions are centrally symmetric, and the constant in John's Theorem~\cite{John} is $\sqrt{d}$ for centrally symmetric bodies. Recalling that $\lambda_0 = 1/(20\sqrt{d})$ we have
\begin{equation}
	M^{4\lambda_0}(x_u) 
		~ \subseteq ~ E(x_u) 
		~ \subseteq ~ M^{4\lambda_0 \sqrt{d}}(x_u)
		~     =     ~ M'(x_u). \label{eq:containment}
\end{equation}
Thus, the generalized cone of rays that intersect $M'(x_u)$ includes all the rays used to define the children of $x_u$. The points $x_v$ that form level $i+1$ of the structure lie on $\partial K(\Delta_{i+1})$ and thus satisfy $\delta(x_v) = \delta(x_u) / 2$. Since by our construction they have disjoint Macbeath regions $M^{\lambda_0}(x_v)$, they constitute a set $Y$ as described in the preconditions of Lemma~\ref{lem:degree}. Each child $v$ of $u$ corresponds to a point $x_v$ such that the ellipsoid $E(x_v)$ intersects the generalized cone. Reasoning as we did above for $x_u$, we have $E(x_v) \subseteq M'(x_v)$. Therefore, the points $x_v$ associated with the children of $u$ constitute a subset of the set $Y'$ given in the lemma. Therefore, the number of children of $x_u$ is $O(1)$, as desired.

\subsection{Storage Space.} \label{ss:storage}

In this section, we show that the total number of nodes in the data structure is $O(1/\eps^{(d-1)/2})$. Since each node has $O(1)$ children, it will follow that the total storage is also $O(1/\eps^{(d-1)/2})$. 

Recall the constants $\Delta_0$ and $\lambda_0 = 1/(20\sqrt{d})$ defined earlier. The number of nodes at level $i$ is bounded above by the cardinality of a maximal set of disjoint Macbeath regions $M^{\lambda_0}(x)$, such that the centers $x$ lie on the boundary of $K(\Delta_i)$, where $\Delta_i = \Delta_0 / 2^i$. Our analysis will make use of the following lemma, which is a straightforward adaptation of Lemma~{3.2}, which is proved in the arXiv version of \cite{AFM16}. 

\begin{lem} \label{lem:bound-mac1}
Let $K \subset \RE^d$ be a convex body in $\gamma$-canonical form. Let $0 < \lambda \le 1/5$ be any fixed constant and let $\Delta \le \gamma/12$ be a real parameter. Let $\mathcal{C}$ be a set of caps, whose widths lie between $\Delta$ and $2\Delta$, such that the Macbeath regions $M^{\lambda}(x)$ centered at the centroids $x$ of the bases of these caps are disjoint. Then $|\mathcal{C}| = O(1/\Delta^{(d-1)/2})$.
\end{lem}

We apply this to bound the number of Macbeath regions that define the nodes of each layer.

\begin{lem} \label{lem:bound-mac}
Let $K \subset \RE^d$ be a convex body in $\gamma$-canonical form for some constant $\gamma$. Let $\Delta_0$ be the constant of Lemma~\ref{lem:mac2-var} and $0 < \lambda \le 1/5$ be any fixed constant. Let $\Delta \le \Delta_0$ be a real parameter. Let $\mathcal{M}$ be a set of disjoint Macbeath regions, each of which has scaling factor $\lambda$ and whose centers lie on the boundary of $K(\Delta)$. Then $|\mathcal{M}| = O(1/\Delta^{(d-1)/2})$.
\end{lem}

\begin{proof}
Let $X$ denote the set of center points of $\mathcal{M}$. By Lemma~\ref{lem:width-delta}, for any $x \in X$, $\width(x)$ is between $\Delta$ and $(2/\gamma)(3d+1) \Delta$. We can partition $X$ (and by extension $\mathcal{M}$) into $O(1)$ groups such that the points in any group have same width to within a factor of two. Let $X'$ denote one of these groups, and let its associated widths be between $w$ and $2w$. Since $\Delta \le \Delta_0$, we have $w \le (2/\gamma)(3d+1) \Delta_0$. Under our assumption that $d \ge 3$, it is easy to verify that the latter quantity does not exceed $\gamma/12$. Thus, the set of caps $C(x)$ for the points of this group satisfy the precondition of Lemma~\ref{lem:bound-mac1}. Applying this lemma yields $|X'| = O(1/w^{(d-1)/2})$. Summing over all the groups, it follows that the total size of $X$ (and hence the number of regions in $\mathcal{M}$) is $O(1/\Delta^{(d-1)/2})$. 
\end{proof}

By Lemma~\ref{lem:bound-mac}, the number of nodes at level $i$ is $O(1/\Delta_i^{(d-1)/2}) = O((2^i / \Delta_0)^{(d-1)/2})$. Recall that $\Delta_0$ depends only on $d$ and $\gamma$, and both $d$ and $\gamma$ are constants. It follows that the number of nodes at level zero is $O(1)$. (This bounds the out-degree of the root node, as alluded to in Section~\ref{ss:children}.) Also, observe that the number of nodes grows geometrically with each level. Therefore, the total number of nodes is dominated by the number of leaves. The leaves are located at level $\ell$, where $\Delta_{\ell}$ is $\Omega(\eps)$. Therefore, the number of leaves, and hence the total number of nodes, is $O(1/\eps^{(d-1)/2})$. 

\subsection{Query Processing.} \label{ss:query}

Finally, let us present the query algorithm for answering $\eps$-approximate ray-shoot\-ing queries. Let $Oq$ denote the query ray. As mentioned earlier, the query algorithm descends the layered DAG structure, visiting a node $u$ at each level such that the associated Macbeath ellipsoid $E(x_u)$ intersects the query ray, until arriving at the leaf level. In order to show that this is well defined, it is necessary to demonstrate that such a node exists at each level of the data structure. Since all the eroded bodies $K(\Delta_i)$ contain the origin, it suffices to show that the union of the Macbeath ellipsoids associated with the nodes of level $i$ cover the boundary of $K(\Delta_i)$. This is established by the following lemma.

\begin{lem} \label{lem:macbeath-cover}
For any $\Delta \le \Delta_0$, let $X$ denote a maximal set of points lying on the boundary of the eroded body $K(\Delta)$ such that the associated Macbeath regions $M^{\lambda_0}(x)$ are pairwise disjoint. Then the collection of Macbeath ellipsoids $\{E(x) ~|~ x \in X\}$ covers $\partial K(\Delta)$.
\end{lem}

\begin{proof}
Consider any point $x' \in \partial K(\Delta)$. Because $X$ is maximal, there must exist $x \in X$ such that $M^{\lambda_0}(x)$ has a nonempty intersection with $M^{\lambda_0}(x')$. By Lemma~\ref{lem:mac-mac}, $M^{\lambda_0}(x') \subseteq M^{4\lambda_0}(x)$. Recalling that $M^{4\lambda_0}(x) \subseteq E(x)$, it follows that $x' \in E(x)$.
\end{proof}

Since $\Delta_i \le \Delta_0$ for each level $i$ of the data structure, it follows from the above lemma that the query procedure will succeed in finding a suitable child for each node visited until it reaches the leaf level. Since each node has a constant number of children, it takes $O(\ell) = O(\log \inv{\eps})$ time to perform this descent.

Recall from Section~\ref{ss:construct} that each leaf node $u$ stores a set $H_u$ of at most $d$ supporting hyperplanes of $K$ whose intersection defines an unbounded simplex that contains $K$ (see Figure~\ref{f:leaf-query}(a)). The query algorithm computes the intersection of the query ray with each of these hyperplanes and returns the closest intersection point $p$ to the origin. The following lemma establishes the correctness of the query processing.

\begin{figure*}[hbtp]
  \centerline{\includegraphics[scale=.75]{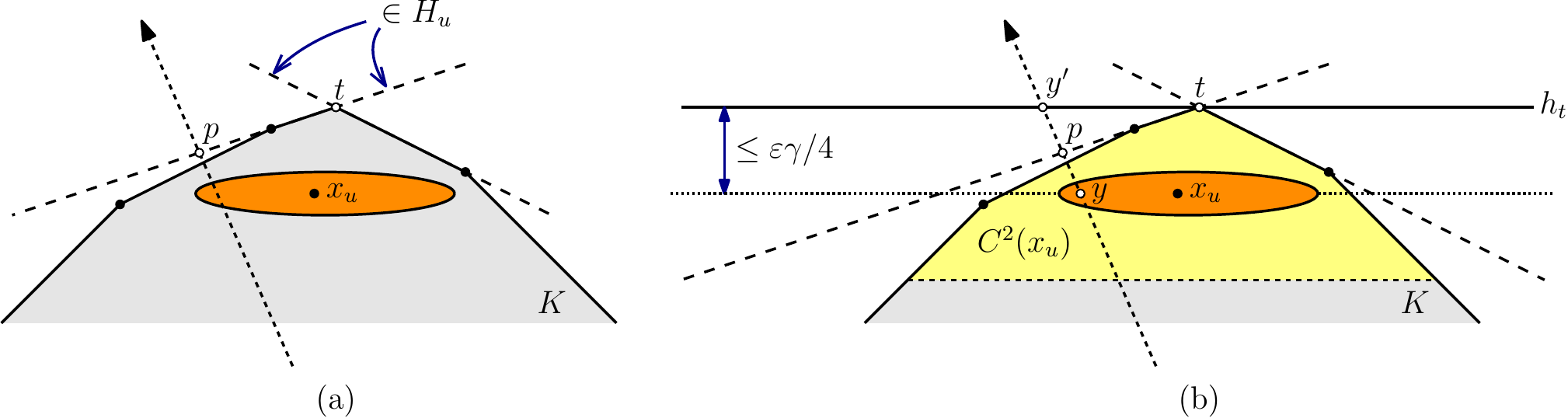}}
  \caption{\label{f:leaf-query}Query processing for a leaf node.}
\end{figure*}

\begin{lem} \label{lem:query-correctness}
Given a query ray $Oq$, the point $p$ returned by the query procedure is a valid answer to the $\eps$-approximate ray-shooting query, and it lies on a supporting hyperplane of $K$.
\end{lem}

\begin{proof}
Observe that $p$ lies at the intersection of the query ray and a supporting hyperplane of $K$. Clearly, $p$ is not internal to $K$, so all that remains is to show that $p$ lies within distance $\eps$ of $K$. Recall that a leaf node $u$ satisfies $\delta(x_u) \leq \gamma^2 \eps / (8 (3d+1))$, and therefore by Lemma~\ref{lem:width-delta}, $\width(x_u) \leq \gamma \eps / 4$. Since the search procedure arrived at node $u$, the ray $Oq$ intersects $E(x_u)$. By Eq.~\eqref{eq:containment} and Lemma~\ref{lem:mac1}, 
\[
  E(x_u) 
	~ \subseteq ~ M'(x_u)
	~ \subseteq ~ C^{6/5}(x_u)
	~ \subseteq ~ C^2(x_u). 
\]
Let $t$ denote the apex of $C^2(x_u)$, and let $h_t$ denote the hyperplane passing through $t$ that is parallel to the base of this cap (see Figure~\ref{f:leaf-query}(b)). By construction, the intersection of the halfspaces $H_u$ associated with $u$ lies within the halfspace bounded by $h_t$ that contains $K$. Let $y$ be any point in $E(x_u) \cap Oq$, and let $y'$ denote the intersection of the ray $Oq$ and $h_t$. By Lemma~\ref{lem:projection}, $\|y'y\| \leq 2 \cdot \width(C^2(x_u)) / \gamma = 4 \cdot \width(x_u)/\gamma \leq \eps$. Therefore $y'$ lies within distance $\eps$ of $K$, implying that $p$ does as well.
\end{proof}

\medskip

Summarizing the results of this section, we have shown that, given a convex polytope $K$ in $\gamma$-canonical form, where $\gamma$ is a constant, and given $\eps > 0$, there exists a data structure that uses $O(1/\eps^{(d-1)/2})$ space and answers $\eps$-approximate ray-shooting queries in time $O(\log \inv{\eps})$. This establishes Lemma~\ref{lem:ray-shoot}, and Theorem~\ref{thm:main} follows immediately. The following lemma justifies our assertion that these bounds are asymptotically optimal.

\begin{lem} \label{lem:opt}
For all sufficiently small $\eps > 0$, any data structure for answering $\eps$-approx\-imate polytope membership queries in $\RE^d$ requires $\Omega(1/\eps^{{(d-1)}/2})$ bits of storage, and if the data structure operates in the decision tree model, the query time is $\Omega(\log \frac{1}{\eps})$ in the worst case.
\end{lem}

\begin{proof}
Consider a Euclidean ball of unit diameter in $\RE^d$, and let $p$ be any point on the boundary of this ball. For any $0 < \eps < \half{1}$, it follows from a simple application of the Pythagorean Theorem that a cap of width $\eps$ whose apex is at $p$ has diameter at most $c \sqrt{\eps}$, for some constant $c$ depending only on $d$. By a simple packing argument there exists a set $P$ of points of size $\Omega((1/\sqrt{\eps})^{d-1}) = \Omega(1/\eps^{(d-1)/2})$ on the boundary of the ball such that the $\eps$-width caps centered at these points are pairwise disjoint. For any two distinct subsets $P'$ and $P''$ of $P$, consider a point $p$ that lies in one subset, say $P'$, but not in the other. It is easy to see that for the query point $p$, the answer to the $\eps$-approximate membership query at $q$ is ``yes'' for $P'$ and ``no'' for $P''$. Therefore, the two data structures for these subsets must differ. It follows that there are $2^{|P|}$ distinct data structures needed to represent the various subsets of $P$. By an information-theoretic argument, such a data structure requires $\Omega(1/\eps^{(d-1)/2})$ bits in the worst case. Assuming that queries are answered in the decision-tree model, such a structure requires depth $\Omega(\log \frac{1}{\eps})$.
\end{proof}

\section{Approximate Nearest-Neighbor Searching} \label{s:ann}

In this section we present a reduction from approximate Euclidean nearest-neighbor searching to approximate polytope membership, or more accurately, to approximate ray-shooting. The reduction is based on the approximate Voronoi diagram (AVD) construction from~\cite{AVD-JACM}. The AVD for an $n$-element point set $X$ employs a height-balanced variant of a quadtree, a balanced box decomposition (BBD) tree~\cite{ARS} to be precise. Each leaf cell $Q$ of the tree stores a set $R \subseteq X$ of \emph{representative points}, which have the property that for any query point $q \in Q$, at least one of these representatives is an $\eps$-nearest neighbor of $q$. We will employ a version of this data structure where the total number of representatives over all the nodes is $O(n \log \inv{\eps})$.

In the data structure of~\cite{AVD-JACM} a query is answered by locating the leaf cell that contains the query point in $O(\log n)$ time, and then selecting the nearest representative from this cell to the query (by simple brute force). Later in~\cite{polytope} it was shown that queries can be answered more efficiently by replacing the brute-force search with an approach based on using approximate polytope membership queries. (This will be discussed below.) These membership queries were applied within the context of a binary search in order to simulate approximate ray shooting. In light of Lemma~\ref{lem:ray-shoot}, we can forgo the binary search, which saves a factor of $O(\log \inv{\eps})$ in the query time.

The approximate ray shooting queries used in \cite{polytope} were of a different nature than those presented here. First, the rays are vertical (parallel to one of the coordinate axes). Second, the hyperplanes near the portion of the polytope's boundary where the ray might hit are not too sharply sloped with respect to the query ray. (More formally, for any $\eps$-approximating convex polytope $P$ of $K$, the angle between the vertical ray and the normal vector of the hyperplane of $P$ hit by this ray is bounded away from $\pi/2$ by a constant.) We refer to this as \emph{vertical slope-restricted approximate ray shooting}. The first part of the following result is proved in~\cite{polytope}, and the slope-restricted variant follows directly by eliminating the binary search.

\begin{lem} \label{lem:reduction}
Let $0 < \eps \leq 1/2$ be a real parameter and $X$ be a set of $n$ points in $\RE^d$. Given a data structure for approximate polytope membership in $d$-dimensional space with query time at most $t_{d}(\eps)$ and storage $s_{d}(\eps)$, it is possible to preprocess $X$ into an $\eps$-approximate nearest neighbor data structure with
\[
	\hbox{Query time:~} O\left(\log n + t_{d+1}(\eps) \cdot \textstyle \log \inv{\eps} \right)
	\quad\hbox{and}
\]
\[
  \textrm{Space:~} O\left(n\,\log\inv{\eps} + n \, \frac{s_{d+1}(\eps)}{t_{d+1}(\eps)}\right).
\]
If vertical slope-restricted approximate ray shooting que\-ries are supported, then the query time is $O(\log n + t_{d+1}(\eps))$. 
\end{lem}

Observe that the space bound varies inversely with the query-time bound. While the query time presented here is $O(\log 1/\eps)$, we can artificially generate higher query times by employing brute-force search. Indeed, this lemma exploits the fact that when brute-force search is used on subsets of size at most $t_{d+1}(\eps)$, the data structure need only be constructed for subsets of size at least $t_{d+1}(\eps)$, of which there are at most $O(n/t_{d+1}(\eps))$.

\subsection{Lifting Transformation.} \label{ss:lifting}

In order to adapt Lem\-ma~\ref{lem:reduction} to our context, we will need to understand a bit more about how it works. It is based on a well-known transformation that maps a point in $\RE^d$ to $\RE^{d+1}$ by projecting it vertically onto a paraboloid. More formally, we can embed a point $p =(x_1, \ldots, x_d)$ in $\RE^d$ into $\RE^{d+1}$ by adding an additional $(d+1)$st coordinate whose value is zero. Let us visualize the $(d+1)$st coordinate axis as being directed vertically upwards. Let $\Psi$ denote the paraboloid $x_{d+1} = \sum_{i=1}^d x_i^2$. Given a point $p \in \RE^d$, the \emph{lifting transformation} projects $p \in \RE^d$ vertically to a point $p^{\uparrow}$ lying on $\Psi$. Define $h(p)$ to be the hyperplane tangent to $\Psi$ at $p^{\uparrow}$, that is,
\[
	h(p) ~ = ~ \bigg\{ (x_1, \ldots, x_{d+1}) ~\bigg|~ x_{d+1} = \sum_{i=1}^d 2 p_i x_i - \|p\|^2 \bigg\}. 
\]
For any $q \in \RE^d$, let $q_{p}$ denote the point of intersection between $h(p)$ and a vertical ray shot upwards from $q$. Letting $\|p q\|$ denote the Euclidean distance between points $p$ and $q$, it is easily verified that $\|q_p q^{\uparrow}\| = \|p q\|^2$. (See \cite{polytope} for details.) 

\begin{figure*}[htbp]
	\centerline{\includegraphics[scale=0.45]{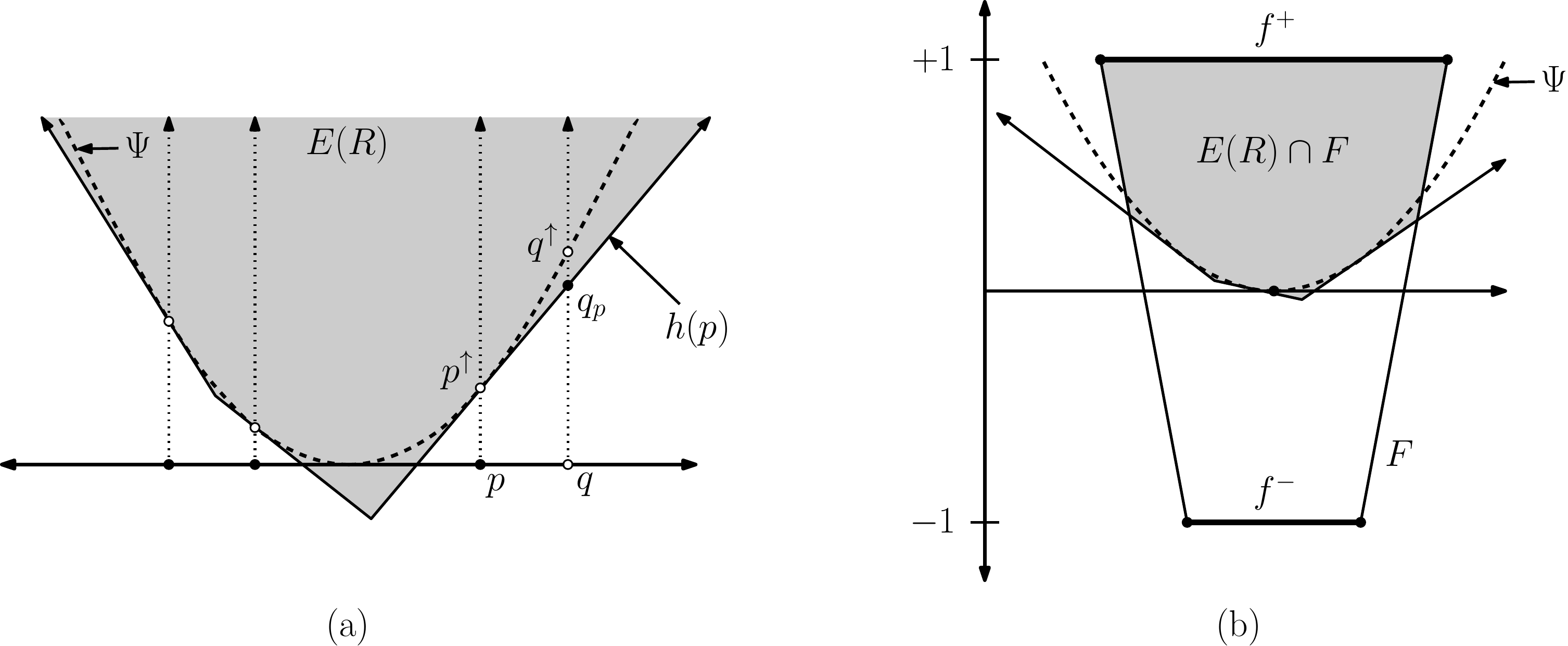}}
	\caption{(a) The lifting transformation and (b) the restriction $E(R) \cap F$.}
	\label{fig:lifting-vertical}
\end{figure*}

Given a finite point set $R$ in $\RE^d$, let $E(R)$ denote the upper envelope of the hyperplanes $h(p)$ for each $p \in R$ (shaded in Figure~\ref{fig:lifting-vertical}(a)). A vertical line through any point $q \in \RE^d$ intersects a facet of $E(R)$. (If the line intersects the boundary between multiple facets, we select one facet arbitrarily.) It follows directly that the nearest neighbor in $R$ of any query point $q$ is the point $p \in R$ whose associated hyperplane $h(p)$ is hit by the vertical line segment passing through $q$. That is, nearest neighbor queries in $\RE^d$ can be reduced to vertical ray-shooting queries against $E(R)$ in $\RE^{d+1}$~\cite{ray-shooting-NN, Edels}. 

While this applies to exact nearest neighbors, it is shown in~\cite{polytope} that the $\eps$-approximate closest representative in $R$ can be determined by simulating vertical ray shooting against a suitable approximation to $E(R)$. In particular, after a normalizing transformation, it can be assumed that the cell $Q$ is centered at the origin, and both $Q$ and the points of $R$ all lie within some constant distance of the origin. The choice of this constant is arbitrary (depending possibly on $d$ but not on $\eps$), and it only affects the constant factors in the query time. We will assume henceforth that this constant is chosen to be $1/2$. 

$E(R)$ is unbounded, and it will be necessary to define a bounded polytope that contains the relevant portion of $E(R)$. Because the distance between any point $q \in \RE^d$ and its closest representative in $R$ is at most one, it follows that for the sake of answering nearest neighbor queries, the relevant portion of $E(R)$ lies within the region bounded by two horizontal hyperplanes $-1 \le x_{d+1} \le +1$. For reasons that will be apparent later, it will be convenient to define this bounded region to be a frustum. Let $f^-$ be the $d$-dimensional hypercube on the hyperplane $x_{d+1} = -1$ satisfying $-1/2 \le x_i \le 1/2$, for $1 \le i \le d$, and let $f^+$ be the $d$-dimensional hypercube on the hyperplane $x_{d+1} = +1$ satisfying $-5/6 \le x_i \le 5/6$, for $1 \le i \le d$ (see Figure~\ref{fig:lifting-vertical}(b)). Let $F$ denote the frustum defined by the convex hull of $f^-$ and $f^+$. Clearly, the relevant portion of $E(R)$ lies within $F$, and so we may restrict attention to the polytope $E(R) \cap F$.

In~\cite{polytope} it is shown that after normalization, answering vertical ray-shooting queries approximately with respect to $E(R)$ is sufficient to answer approximate nearest neighbor queries with respect to $R$. The following lemma restates this result in a manner that is suitable for our context. The proof follows directly from the analysis of~\cite{polytope}, but with the constant factors adjusted accordingly. 

\begin{lem} \label{lem:absolute-reduction}
Given an AVD cell $Q$ and representative set $R$ that have been normalized as specified above, there exists a positive constant $c$ (depending possibly on $d$ but not on $\eps$) such that following holds. Let $R'$ be any subset of $R$ such that the Hausdorff distance between $E(R') \cap F$ and $E(R) \cap F$ is at most $\eps/c$. Then for any $q \in Q$, if $p' \in R'$ is the defining point of the facet of $E(R')$ that is hit by a vertical line through $q$, then $p'$ is an $\eps$-approximate nearest neighbor of $q$ within $R$.
\end{lem}

\subsection{From Vertical to Central Ray Shooting.} \label{ss:central}

The principal impediment to applying this result to the polytope membership data structure described in Section~\ref{s:ds} is that the ray-shooting used in Lemma~\ref{lem:absolute-reduction} is vertical, and here it is targeted towards a point at the center of the polytope. In the remainder of this section we will show how to adapt vertical ray shooting to central ray shooting. Our approach involves defining a projective transformation that maps vertical lines to lines passing through a given point. 

Before giving the transformation, let us recall some basic facts from projective geometry and homogeneous coordinates. A point $p = (x_1, \ldots, x_{d+1}) \in \RE^{d+1}$ can be represented using homogeneous coordinates as a $(d+2)$-vector $[x_0, x_1, \ldots, x_{d+1}]$, where $x_0 = 1$. (We use square brackets for homogeneous coordinates and parentheses for Cartesian coordinates.) Two nonzero homogeneous vectors represent the same point in space if they are equal up to a nonzero scale factor. The point at infinity in the direction given by the nonzero vector $(x_1, \ldots, x_{d+1})$ is represented by the homogeneous coordinates $[0, x_1, \ldots, x_{d+1}]$. Any projective transformation can be defined by applying a linear transformation to the homogeneous coordinates followed by a normalization step in which all the coordinates are divided by the $x_0$ coordinate (assuming that it is nonzero).

Given a point $p = [x_0, x_1, \ldots, x_{d+1}]$, consider the projective transformation
\begin{eqnarray*}
  T(p)
	&    =   & [ 4 x_0 + x_{d+1}, 4 x_1, \ldots, 4 x_d, 2 x_{d+1} ] \\
	& \equiv & \left( \frac{4 x_1}{4 + x_{d+1}}, \ldots, \frac{4 x_d}{4 + x_{d+1}}, \frac{2 x_{d+1}}{4 + x_{d+1}} \right).
\end{eqnarray*}

Let $S$ denote a hypersphere of unit radius that is centered one unit above the origin. Let $p_0 = (0, \ldots, 0, 2)$ denote the topmost point of $S$. The following lemma states the important properties of $T$ for our purposes.

\begin{lem} \label{lem:transform-props}
The projective transformation $T$ satisfies the following:
\begin{enumerate}
\item[$(1)$] $T$ maps horizontal hyperplanes to horizontal hyperplanes and it fixes the hyperplane $x_{d+1} = 0$, that is, for any $p \in \RE^d$, $T(p) = p$.

\item[$(2)$] $T$ maps the point at vertical infinity $($having homo\-geneous coordinates $[0, \ldots, 0, 1]${}$)$ to $p_0$ $($having the homogeneous coordinates $[1, 0, \ldots, 0, 2]${}$)$. Therefore, the vertical line through any point $p \in \RE^d$ is mapped to the line $\overline{p p_0}$ (see Figure~\ref{fig:lifting-central}(a)).

\item[$(3)$] If $x_{d+1} > -4$, then $T$ preserves the signs of the coor\-dinates of the transformed point.

\item[$(4)$] $T$ maps $F$ to an axis-aligned hyperrectangle whose vertical projection is a hypercube of side length $4/3$ centered at the origin and whose vertical extent is $-2/3 \le x_{d+1} \le 2/5$ (see Figure~\ref{fig:lifting-central}(b)).

\item[$(5)$] $T$ maps the paraboloid $\Psi$ to the punctured sphere $S \setminus \{p_0\}$. Therefore, for any $p \in \RE^d$, $T(p^{\uparrow})$ is the intersection of the line $\overline{p p_0}$ and $S \setminus \{p_0\}$. Because projective transformations preserve flatness, $T(h(p))$ is the hyperplane tangent to $S$ at this point.

\item[$(6)$] The inverse of $T$ is
\begin{eqnarray*}
  T^{-1}(p)
	~    =   ~ \frac{1}{8} [ 2 x_0 - x_{d+1}, 2 x_1, \ldots, 2 x_d, 4 x_{d+1} ] \\
	~ \equiv ~ \left( \frac{2 x_1}{2 - x_{d+1}}, \ldots, \frac{2 x_d}{2 - x_{d+1}}, \frac{4 x_{d+1}}{2 - x_{d+1}} \right).
\end{eqnarray*}
\end{enumerate}
\end{lem}

\begin{figure*}[htbp]
	\centerline{\includegraphics[scale=0.45]{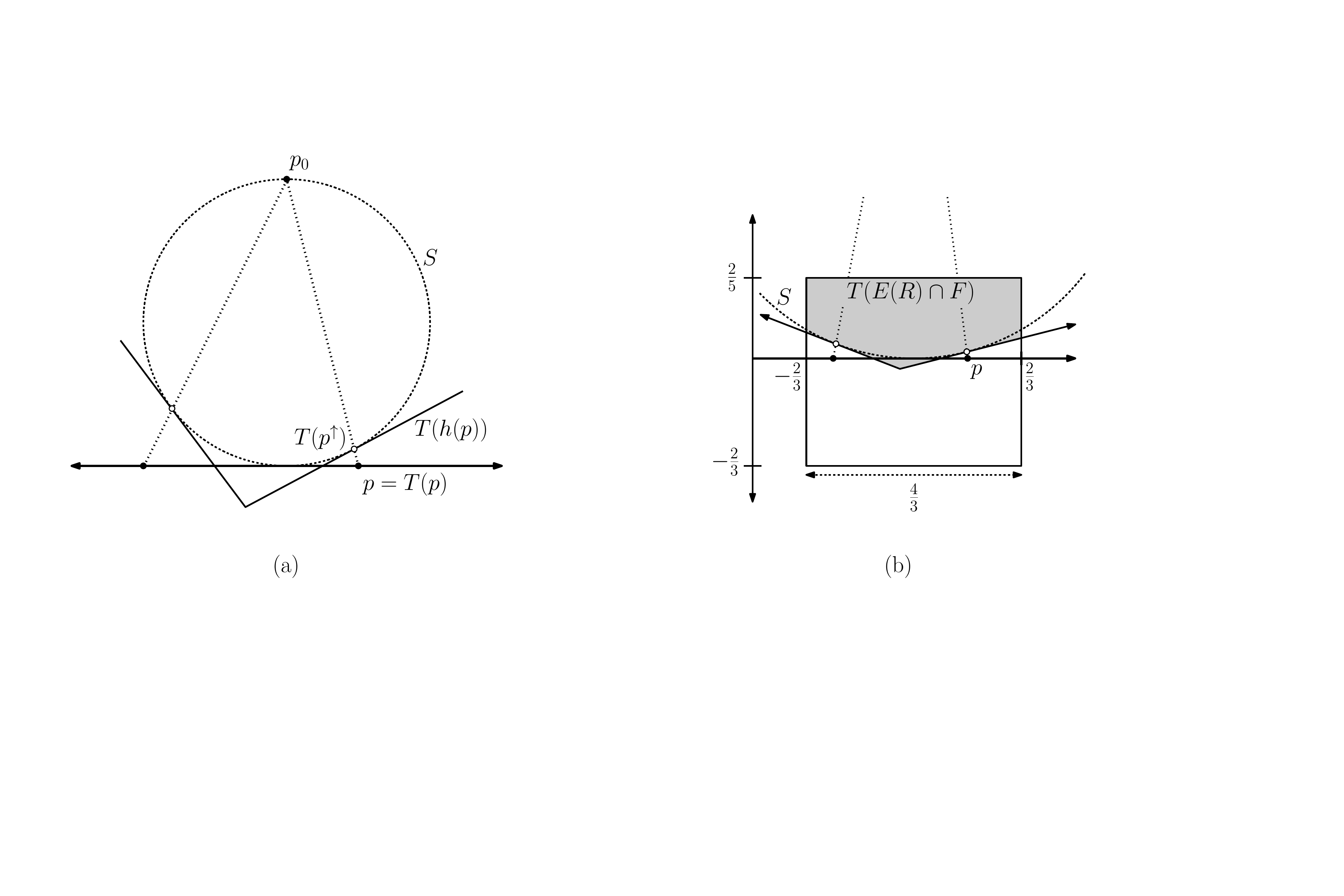}}
	\caption{The projective transformation and Lemma~\ref{lem:transform-props}.}
	\label{fig:lifting-central}
\end{figure*}

\begin{proof}
Assertions~(1)--(3) and~(6) are straightforward to verify. Assertion~(4) can be verified by transforming the corner points of $F$, $(\pm 1/2, \ldots, \pm 1/2, -1)$ and $(\pm 5/6, \ldots, \pm 5/6, 1)$. To see assertion~(5), observe that the points of $S$ can be described as the zero set of the function
\[
	\varphi(x_1, \ldots, x_{d+1}) 
		~ = ~ \sum_{i = 1}^d x_i^2 + (x_{d+1} - 1)^2 - 1. 
\]
Let $p = [1, x_1, \ldots, x_{d+1}] \equiv (x_1, \ldots, x_{d+1})$ denote the coordinates of any point on $\Psi$. Letting $\sigma = \sum_{i = 1}^d x_i^2$, we have $x_{d+1} = \sigma$. Applying $T$ yields
\begin{eqnarray*}
	T(p) 
		&   =   & \left[ 4 + \sigma, 4 x_1, \ldots, 4 x_d, 2 \sigma \right]\\
		& \equiv & \left( \frac{4 x_1}{4 + \sigma}, \ldots, \frac{4 x_d}{4 + \sigma}, \frac{2 \sigma}{4 + \sigma} \right).
\end{eqnarray*}
It is straightforward to verify that $\varphi(T(p)) = 0$ and $\lim_{\sigma \rightarrow \infty} T(p) = p_0$.
\end{proof}

Assertion~(4) is the reason for defining $F$ in the manner that we did. Projective transformations preserve flatness, and hence $T(E(R) \cap F)$ is a polytope. It follows that for any $q \in Q$, we can compute its exact nearest neighbor in $R$ by determining the lower facet of $T(E(R) \cap F)$ that is hit by the line $\overleftrightarrow{q p_0}$. (There is an obvious connection with the relationship observed by Brown~\cite{Bro79} between the stereographic projection and the Voronoi diagram.) 

\subsection{Preserving Distances.} \label{ss:low-distortion}

In order to show that this transformation can be used for approximate nearest neighbor searching, we show that $T$ does not significantly distort the distances between points of interest. In particular, we show that if two points of $T(F)$ are close then their preimages are also close.

\begin{lem} \label{lem:distortion}
There exists a constant $c'$ such that for any two points $p,q \in T(F)$ such that $\|p q\| \le 1/4$, $\|T^{-1}(p) T^{-1}(q)\| \le c' \|p q\|$.
\end{lem}

\begin{proof}
Let $\|v\|_{\infty}$ denote the $L_{\infty}$ length of a vector $v$. Consider any two points $p,q \in T(F)$ such that $\|p q\| \le 1/4$. We can express $q$ as $p + \vec{\delta}$, where $\|\vec{\delta}\|_{\infty} \le 1/4$. Let $\delta_\infty = \|\vec{\delta}\|_{\infty}$. We will show that the $L_{\infty}$ distance between $T^{-1}(p)$ and $T^{-1}(p + \vec{\delta})$ is at most $c'' \kern+1pt \delta_\infty$ for some constant $c''$. It will follow that $T^{-1}$ increases Euclidean distances for the points of interest by a factor of at most $c' = c'' (d+1)$.

In order to establish the above assertion, let $p = (x_1, \ldots, x_{d+1})$ and let $\vec{\delta} = (\delta_1, \ldots, \delta_{d+1})$. We begin with the following easy inequalities. Given $1 \le i \le d+1$, by our bounds on $p$ and $\vec{\delta}$ we have
\begin{equation}
\begin{split}
	(i)& \; |x_{i} + \delta_i| ~ < ~ 1
		\qquad\\
	(ii)& \; 2 - x_{d+1} ~ > ~ 1
		\qquad \\
	(iii)& \; 2 - x_{d+1} - \delta_{d+1} ~  >  ~ 1. \label{eqn:distortion}
\end{split}
\end{equation}
(The worst case for the first inequality arises when $x_i = 2/3$ and $\delta_i = 1/4$, and the worst case for the second and third inequalities occur when $x_{d+1} = 2/3$ and $\delta_{d+1} = 1/4$.) We will also make use of the identity $a/(b-c) = a/b + a c/b(b-c)$, assuming $b$ and $b - c$ are both nonzero. 

Consider the transformed point $T^{-1}(q) = T^{-1}(p + \vec{\delta})$. By applying Lemma~\ref{lem:transform-props}(6) and the above identity, for $1 \le i \le d$, we find that the $i$th coordinate is mapped to
\begin{eqnarray*}
	\frac{2 (x_i + \delta_i)}{2 - (x_{d+1} + \delta_{d+1})}
		&  =  & \frac{2 (x_i + \delta_i)}{(2 - x_{d+1}) - \delta_{d+1}} \\
		&  =  & \frac{2 (x_i + \delta_i)}{2 - x_{d+1}} + \frac{2 (x_i + \delta_i) \delta_{d+1}}{(2 - x_{d+1})(2 - x_{d+1} - \delta_{d+1})} \\
		&  =  & \frac{2 x_i}{2 - x_{d+1}} + \frac{2 \delta_i}{2 - x_{d+1}} + \frac{2 (x_i + \delta_i) \delta_{d+1}}{(2 - x_{d+1})(2 - x_{d+1} - \delta_{d+1})}.
\end{eqnarray*}
After some expansion, this is equal to
\[
	\frac{2 x_i}{2 - x_{d+1}} + \frac{2 \delta_i}{2 - x_{d+1}} + \frac{2 (x_i + \delta_i) \delta_{d+1}}{(2 - x_{d+1})(2 - x_{d+1} - \delta_{d+1})}.
\]
The first term is the $i$th coordinate of $T^{-1}(p)$. By Eq.~\eqref{eqn:distortion}, the second term has absolute value at most $2 \delta_{\infty}$. The third term has absolute value at most $2 \delta_{\infty}$. Therefore the $i$th coordinate of $T^{-1}(p + \vec{\delta})$ is within distance $4 \delta_{\infty}$ of the corresponding coordinate of $T^{-1}(p)$. By applying a similar analysis to the $(d+1)$st coordinate of $T^{-1}(p + \vec{\delta})$, it follows that this coordinate is within distance $8 \delta_{\infty}$ of the corresponding coordinate of $T^{-1}(p)$.
Therefore, by setting $c'' = 8$, it follows that the $L_{\infty}$ distance between $T^{-1}(p)$ and $T^{-1}(p + \vec{\delta})$ is at most $c'' \delta_{\infty}$.
\end{proof}

By combining Lemmas~\ref{lem:absolute-reduction} and~\ref{lem:distortion}, it follows that in order to answer $\eps$-approximate nearest neighbor queries in $\RE^d$ for an AVD leaf cell $Q$ and a set $R$ of representatives, it suffices to first apply the normalizing transformation to $Q$, construct an approximate ray-shooting data structure for $T(E(R) \cap F)$ that answers queries to within an absolute error of $\eps' = \eps/c c'$, where $c$ and $c'$ are the constant factors of these respective lemmas. It follows from Lemma~\ref{lem:distortion} that the result is an absolute $(\eps/c)$-approximation to the corresponding vertical ray shooting query in $T^{-1}(T(E(R) \cap F)) = E(R) \cap F$. From Lemma~\ref{lem:absolute-reduction} such an approximation suffices to answer $\eps$-approximate nearest neighbor queries.

In order to apply the results of Section~\ref{s:ds}, we require that the polytope in question be in $\gamma$-canonical form for a suitable constant $\gamma$ and that rays be directed towards the origin. To do this, we modify $T(E(R) \cap F)$. Recall that it is contained within an axis-aligned hyperrectangle whose topmost facet is at $x_{d+1} = 2/5$. We move this topmost facet up to $x_{d+1} = 8/3$ (see Figure~\ref{fig:lifting-modified}). The point $p_0$ lies within the interior of the modified polytope. Further a ball of radius $2/3$ centered at $p_0$ is contained entirely within this polytope, and the polytope is completely contained within a ball of radius less than $3 + d$. By translating this modified polytope so that $p_0$ coincides with the origin, the result is in $\gamma$-canonical form for $\gamma > 2/3(3+d)$.

\begin{figure*}[htbp]
	\centerline{\includegraphics[scale=0.45]{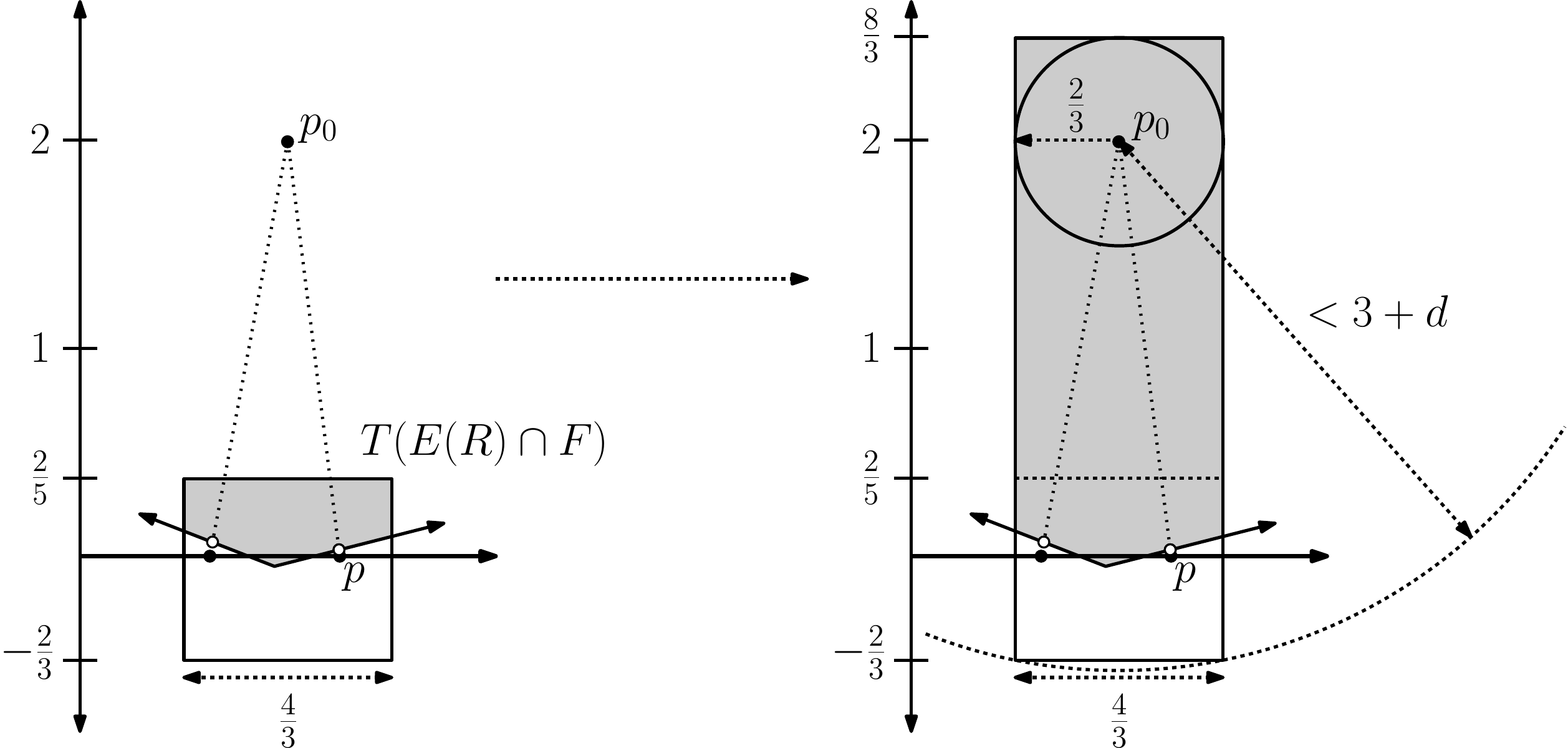}}
	\caption{The modified polytope.}
	\label{fig:lifting-modified}
\end{figure*}

Now, the data structure described in Section~\ref{s:ds} can be applied to the modified polytope. The witness hyperplane (as described in Section~\ref{ss:construct}) that is hit by the ray provides the identity of the desired nearest-neighbor representative, that is, the approximate nearest neighbor of the query point.

Given a parameter $m$ such that 
\[
	\log \inv{\eps} 
		~ \leq ~ m 
		~ \leq ~ \frac{1}{\eps^{d/2}\log \inv{\eps}},
\]
we set $t_{d+1}(\eps) = 1/(m \cdot \eps^{d/2})$ and $s_{d+1}(\eps) = 1/\eps^{d/2}$. Note that for $m$ in this range we have $t_{d+1}(\eps) \ge \log \inv{\eps}$ and so our data structure can achieve this query time for $\eps$-approximate ray-shooting queries. By the results of this section, these bounds apply to vertical slope-restricted approximate ray shooting queries as well. By applying Lemma~\ref{lem:reduction} we obtain a data structure for approximate Euclidean nearest-neighbor searching with query time $O(\log n + 1/(m \cdot \eps^{d/2}))$ and space $O(n \log \inv{\eps} + n m) = O(n m)$. This establishes Theorem~\ref{thm:ann-ub}.


\pdfbookmark[1]{References}{s:ref}
\bibliographystyle{abbrv}
\bibliography{shortcuts,polytope}

\appendix
\section{Appendix} \label{s:appendix}

For the sake of completeness, we give a proof of Lemma~\ref{lem:mac-mac}.


{\MacMacLemmaStmt*}

\begin{proof}
Let $z$ be a point in the intersection of $M^{\lambda}(x)$ and $M^{\lambda}(y)$. Then we can write $z$ as:
\[
	z = x + \lambda (x - p_1) = y + \lambda (p_2 - y),
\]
where $p_1,p_2 \in K$. Equating the two expressions for $z$ above, we obtain
\[
	y = \frac{ (1 + \lambda) x - \lambda p_1 - \lambda p_2 }{1 - \lambda}.
\]
Consider any point $w \in M^{\lambda}(y)$. We have
\[
	w = y + \lambda(y - p_3) = (1+\lambda) y - \lambda p_3,
\]
where $p_3 \in K$. Substituting the expression obtained above for $y$, we have
\[
	w = \frac{(1+\lambda) ((1 + \lambda) x - \lambda p_1 - \lambda p_2)}{1-\lambda} - \lambda p_3,
\]
which simplifies to
\[
	w = x + \frac{\lambda (3 + \lambda)}{1 - \lambda} (x - p),
\]
where 
\[
	p = \frac{1+\lambda}{3+\lambda} p_1 + \frac{1+\lambda}{3+\lambda} p_2 + \frac{1-\lambda}{3+\lambda} p_3.
\]
As $p$ is a convex combination of $p_1,p_2$ and $p_3$, $p \in K$. Thus, we have shown that 
\begin{equation}
	M^{\lambda}(y) 
		~ \subseteq ~ x + \frac{\lambda(3+\lambda)}{1-\lambda} (x-K). \label{eq:XminusK}
\end{equation}

In an analogous manner, we next show that 
\begin{equation}
	M^{\lambda}(y) 
		~ \subseteq ~ x + \frac{\lambda(3+\lambda)}{1-\lambda} (K-x). \label{eq:KminusX}
\end{equation}
Again, let $z$ be any point in the intersection of $M^{\lambda}(x)$ and $M^{\lambda}(y)$. We can write $z$ as:
\[
	z = x + \lambda (k'_1 - x) = y + \lambda (y - k'_2),
\]
where $k'_1,k'_2 \in K$. Equating the two expressions for $z$ above, we obtain
\[
	y = \frac{ (1 - \lambda) x + \lambda k'_1 + \lambda k'_2 }{1 + \lambda}.
\]
Consider any point $w \in M^{\lambda}(y)$. We have
\[
	w = y + \lambda(k'_3 - y) = (1-\lambda) y + \lambda k'_3,
\]
where $k'_3 \in K$.
Substituting the expression obtained above for $y$, we have
\[
	w = \frac{(1-\lambda) ((1 - \lambda) x + \lambda k'_1 + \lambda k'_2)}{1+\lambda} + \lambda k'_3,
\]
which simplifies to
\[
	w = x + \frac{\lambda (3 - \lambda)}{1 + \lambda} (p' - x),
\]
where 
\[
	p' = \frac{1-\lambda}{3-\lambda} k'_1 + \frac{1-\lambda}{3-\lambda} k'_2 + \frac{1+\lambda}{3-\lambda} k'_3.
\]
As $p'$ is a convex combination of $k'_1,k'_2$ and $k'_3$, $p' \in K$. Letting $p''$ denote the point on segment $xp'$ such that
\[
	\frac{\lambda (3-\lambda)}{1+\lambda} (p'-x) = \frac{\lambda (3+\lambda)}{1-\lambda} (p''-x),
\]
we can write 
\[
	w = x + \frac{\lambda (3 + \lambda)}{1 - \lambda} (p'' - x),
\]
where $p'' \in K$. Thus, 
\[
	M^{\lambda}(y) 
		~ \subseteq ~ x + \frac{\lambda(3+\lambda)}{1-\lambda} (K-x),
\]
which establishes Eq.~(\ref{eq:KminusX}). By combining this with Eq.~(\ref{eq:XminusK}), we obtain $M^{\lambda}(y) \subseteq M(x,\lambda(3+\lambda)/(1-\lambda))$. Since $\lambda \le 1/5$, it is easy to see that $(3+\lambda)/(1-\lambda) \le 4$. Thus $M^{\lambda}(y) \subseteq M(x,4\lambda)$, completing the proof.
\end{proof}

\end{document}